\RequirePackage[l2tabu, orthodox]{nag}
\documentclass[10pt,a4paper,onecolumn]{article}
\usepackage[margin=2cm]{geometry}
\usepackage{amsmath,amssymb,amsfonts,amsthm,thmtools,xcolor,hyperref,cite,mathtools,amssymb}
\usepackage[capitalise, noabbrev]{cleveref}
\hypersetup{citecolor=violet,linkcolor=blue,urlcolor=violet,colorlinks=true,pdfborder={0 0 0},linktocpage=true}
\usepackage{multirow}

\theoremstyle{plain}
\newtheorem{thm}{Theorem}
\newtheorem{lem}{Lemma}
\newtheorem{prop}[lem]{Proposition}
\newtheorem{cor}[lem]{Corollary}
\newtheorem{remark}{Remark}

\theoremstyle{definition}
\newtheorem{defn}[lem]{Definition}

\newcommand{\braket}[2]   {\left\langle #1\middle|#2\right\rangle}
\newcommand{\ketbra}[2]   {\left| #1\middle\rangle\middle\langle#2\right|}
\newcommand{\abs}[1]      {\left|#1\right|}
\newcommand{\norm}[1] 	  {\left\|#1\right\|}

\newcommand{\eps}{\varepsilon}
\newcommand{\cX}{\mathcal{X}}
\newcommand{\cS}{\mathcal{S}}
\newcommand{\cSo}{\mathcal{S}_{\circ}}
\newcommand{\cW}{\mathcal{W}}
\newcommand{\cH}{\mathcal{H}}
\newcommand{\cI}{\mathcal{I}}
\newcommand{\cP}{\mathcal{P}}

\newcommand{\im}{\mathop{\mathrm{im}}}

\renewcommand{\epsilon}{\eps}

\DeclareMathOperator{\Tr} {Tr}
\DeclareMathOperator{\Var}{Var}

\usepackage[affil-it]{authblk}
\title{Moderate deviation analysis for classical communication over quantum channels}
\author[1]{Christopher T.~Chubb\thanks{\href{mailto:christopher.chubb@sydney.edu.au}{christopher.chubb@sydney.edu.au}}}
\author[2,3]{Vincent Y.~F.~Tan\thanks{\href{mailto:vtan@nus.edu.sg}{vtan@nus.edu.sg}}}
\author[1,4]{Marco Tomamichel\thanks{\href{mailto:marco.tomamichel@uts.edu.au}{marco.tomamichel@uts.edu.au}}}
\affil[1]{Centre for Engineered Quantum Systems, School of Physics, University of Sydney, Sydney, Australia.
}
\affil[2]{Department of Electrical and Computer Engineering, National University of Singapore, Singapore.}
\affil[3]{Department of Mathematics, National University of Singapore, Singapore.}
\affil[4]{Centre for Quantum Software and Information, University of Technology Sydney, Sydney, Australia.}

\begin{document}

\maketitle


\begin{abstract}
	We analyse families of codes for classical data transmission over quantum channels that have both a vanishing probability of error and a code rate approaching capacity as the code length increases. To characterise the fundamental tradeoff between decoding error, code rate and code length for such codes we introduce a quantum generalisation of the moderate deviation analysis proposed by Alt\u{u}g and Wagner as well as Polyanskiy and Verd\'u. We derive such a tradeoff for classical-quantum (as well as image-additive) channels in terms of the channel capacity and the channel dispersion, giving further evidence that the latter quantity characterises the necessary backoff from capacity when transmitting finite blocks of classical data. To derive these results we also study asymmetric binary quantum hypothesis testing in the moderate deviations regime. Due to the central importance of the latter task, we expect that our techniques will find further applications in the analysis of other quantum information processing tasks.
\end{abstract}

\section{Introduction}

The goal of information theory is to find the fundamental limits imposed on information processing and transmission by the laws of physics. One of the early breakthroughs in quantum information theory was the characterisation of the capacity of a classical-quantum (c-q) channel to transmit classical information by Holevo~\cite{holevo98,holevo73b} and Schumacher--Westmoreland~\cite{schumacher97}. The \emph{classical capacity} of a quantum channel is defined as the maximal rate (in bits per channel use) at which we can transmit information such that the decoding error vanishes asymptotically as the length of the code increases. However, for many practical applications there are natural restrictions on the code length imposed, for example, by limitations on how much quantum information can be processed coherently. Therefore it is crucial to go beyond the asymptotic treatment and understand the intricate tradeoff between decoding error {probability}, code rate and code length. 

For this purpose, we will study families of codes that have {\em both} a rate approaching the capacity and an {error probability} that vanishes asymptotically as the code length $n$ increases. The following tradeoff relation gives a rough illustration of our main result: if the code rate approaches capacity as $\Theta(n^{-t})$ for some $t \in (0,1/2)$, then the decoding error cannot be smaller than $\exp(- \Theta(n^{1-2t}))$. In fact, we will show that the constants implicit in the $\Theta$ notation are determined by a second channel parameter {beyond} the capacity, called the \emph{channel dispersion}. We will also show that this relation is tight, i.e., there exist families of codes achieving equality {asymptotically}. 

Our work thus complements previous work on the boundary cases corresponding to {$t  \in \{0, 1/2\}$}. The error exponent (or reliability function) of c-q channels (see, e.g., Refs.~\cite{holevo00,hayashi07,dalai13}) corresponds to the case $t = 0$ where the rate is bounded away from capacity and the error probability vanishes exponentially in $n$. This is also called the \emph{large deviations} regime.  Moreover, the second-order asymptotics of c-q channels were evaluated by Tomamichel and Tan~\cite{tomamicheltan14}. They correspond to the case $t = 1/2$ where the rate approaches capacity as $\Theta(n^{-1/2})$ and the  error probability  is non-vanishing. This is also called the \emph{small deviations} regime. 

In the present work, we consider the entire regime in between, which is dubbed the \emph{moderate deviation} regime.\footnote{In the technical analysis, we are considering moderate deviations from the mean of a sum of independent  {log-likelihood} ratios, thus justifying the name emanating from {statistics~\cite[Theorem~3.7.1]{dembo98}}.} The different parameter regimes are illustrated in Fig.~\ref{fig:regimes}.

\begin{figure}[t!]
	\centering
	\includegraphics{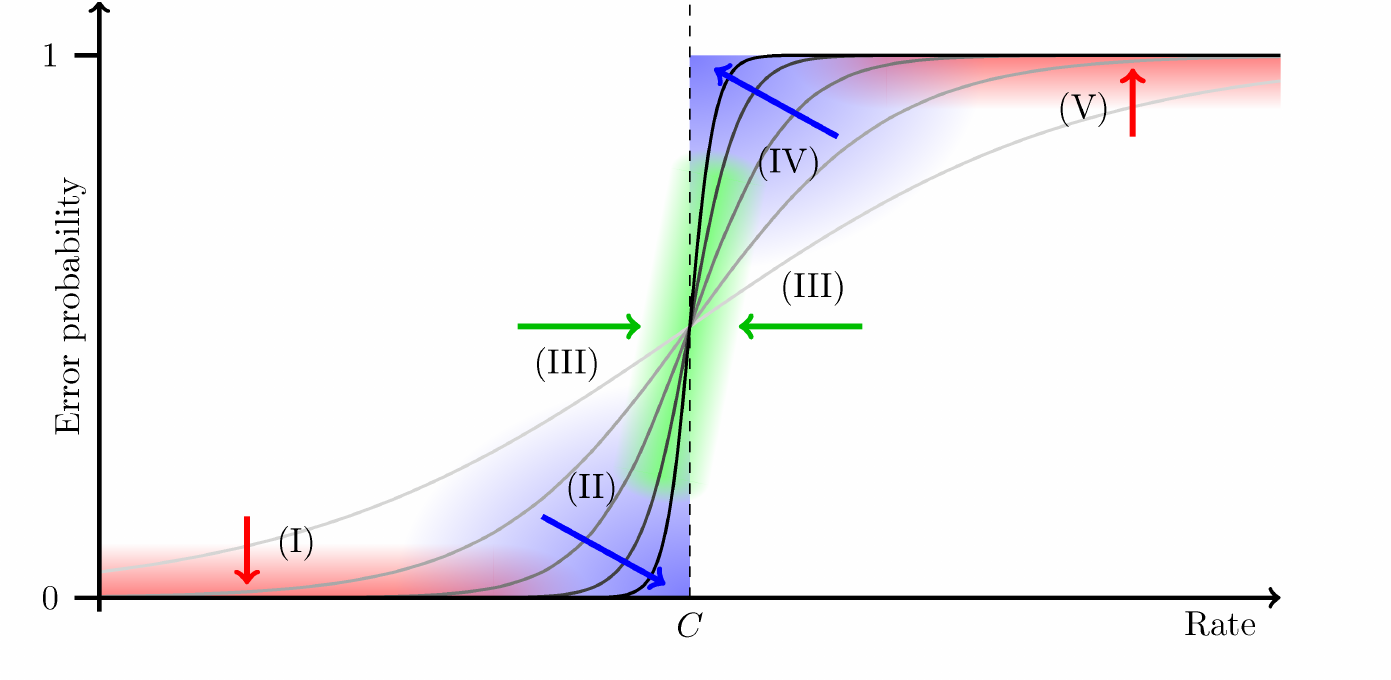}
	\begin{tabular}{|l|c|c|c|c|c|}
		\hline
		& (I) & (II) & (III) & (IV) & (V) \\
		\hline
		\multirow{2}{*}{regime} & error  & \!moderate deviation\! & constant error & moderate deviation & strong converse \\
		& exponent & (below capacity) & (second-order) & (above capacity) & exponent \\
		\hline
		error prob.\ & \footnotesize $\!\exp(-\Theta(n))\!$ & \footnotesize $\exp(-o(n))$ \& $\omega(1)$ & \footnotesize $\Theta(1)$ & \footnotesize \!$1 - \exp(-o(n))$ \& \footnotesize $1 - \omega(1)$\! & \footnotesize $1 - \exp(-\Theta(n))$ \\
		\hline
		code rate & \footnotesize $C - \Theta(1)$ & \footnotesize $C - o(1)$ \& $C - \omega\big(n^{-\frac12}\big)$ & \footnotesize $C - \Theta\big(n^{-\frac12}\big)$ & \footnotesize $C + o(1)$ \& $C + \omega\big(n^{-\frac12}\big)$ & \footnotesize $C + \Theta(1)$ \\
		\hline
	\end{tabular}
	\caption{The figure shows the optimal error probability as a function of the rate, for different block lengths. Darker lines correspond to longer block lengths, and the capacity is denoted by $C$. The table shows the asymptotics in each region, as the blocklength $n$ goes to infinity. The functions of $n$ implicit in the $\Theta$, $o$, and $\omega$ notation are assumed to be positive-valued.}
	\label{fig:regimes}
\end{figure}

\paragraph*{Main results.} 
Before we present our main results, let us introduce the notion of a 
\emph{moderate sequence} of real numbers, $\{ x_n\}_n$ for $n \in \mathbb{N}$, 
whose defining properties are that $x_n \searrow 0$ and $\sqrt{n}\, x_n \to 
+\infty$ as $n \to \infty$.\footnote{As mentioned above an archetypical 
	moderate sequence is $x_n = \Theta(n^{-t})$ for some $t \in (0, \frac12)$. The 
	boundary cases are not included\,---\,in fact $t = 0$ requires a large 
	deviation analysis whereas $t = \frac12$ requires a small deviation analysis.} 
Our two main results concern binary asymmetric quantum hypothesis testing and 
c-q channel coding.

\begin{enumerate}
	\item The first result, presented in detail in Sect.~\ref{sec:hypo}, concerns binary quantum hypothesis testing between a pair of quantum states $\rho$ and~$\sigma$. We show that for any moderate sequence $x_n$, there exists a sequence of tests $\{ Q_n \}_n$ such that the two kinds of errors satisfy
	\begin{align}
		\Tr \rho^{\otimes n} (1 - Q_n) = e^{-nx_n^2}  &~\textrm{and}~ 
		\Tr \sigma^{\otimes n} Q_n = \exp\Big(-n \Big( D(\rho\|\sigma) - \sqrt{2 V(\rho\|\sigma) }\, x_n  + o(x_n) \Big)\Big) \,,
		\intertext{and another sequence of tests $\{ Q_n' \}_n$ such that the errors satisfy
		}
		\Tr \rho^{\otimes n} (1-Q_n')  = 1-e^{- n x_n^2} &~\textrm{and}~  
		\Tr \sigma^{\otimes n} Q_n' = \exp\Big(-n \Big( D(\rho\|\sigma) + \sqrt{2 V(\rho\|\sigma) }\, x_n  + o(x_n) \Big)\Big) \,,
	\end{align}
	where $D(\cdot\|\cdot)$ and $V(\cdot\|\cdot)$ denote the relative entropy~\cite{umegaki62} and relative entropy variance~\cite{tomamichel12,li12}, respectively. (The reader is referred to the next section for formal definitions of all concepts discussed here.)
	Most importantly, we show that both of these tradeoffs are in fact optimal.
	
	\item The main result, covered in Sect.~\ref{sec:channels}, concerns coding over a memoryless classical-quantum channel $\cW$. Let us denote by $M^*(\cW;n,\eps)$ the maximum $M \in \mathbb{N}$ such that there exists a code transmitting one out of $M$ messages over $n$ uses of the channel $\cW$ such that the average probability of error does not exceed $\eps$. For any sequence of tolerated error probabilities $\{\eps_n\}_n$ vanishing sub-exponentially with $\eps_n = e^{-nx_n^2}$, we find that
	\begin{align}
		\label{eq:moderatesecondorder1}
		\frac{1}{n} \log M^*(\cW;n, \eps_n) &= C(\cW) - \sqrt{2 V_{\min}(\cW)}\, x_n + o(x_n) \,,\\
		\label{eq:moderatesecondorder2}
		\frac{1}{n} \log M^*(\cW;n, 1-\eps_n) &= C(\cW) + \sqrt{2 V_{\max}(\cW)}\, x_n + o(x_n) \,,
	\end{align}
	where $C(\cdot)$ denotes the channel capacity and~$V_{\min}(\cdot)$ and $V_{\max}(\cdot)$ denote the minimal and maximal channel dispersion as defined in~Ref.~\cite{tomamicheltan14}, respectively.
	This result holds very generally for channels with arbitrary input alphabet and without restriction on the channel dispersion, strengthening also the best known results for classical channels. Moreover, as in~Ref.~\cite{tomamicheltan14}, this generality allows us to lift the above result to a statement about coding classical information over image-additive quantum channels and general channels as long as the encoders are restricted to prepare separable states.
\end{enumerate}

Since quantum hypothesis testing underlies many other quantum information processing tasks such as entanglement-assisted classical communication as well as private and quantum communication, we expect that our techniques will have further applications in quantum information theory.

\paragraph*{Related work.}

For classical channels, Alt\u{u}g and Wagner~\cite{altug14} first established the best decay rate of the average error probability for a class of discrete memoryless channels (DMCs) when the code rate approaches capacity at a rate slower than $\Theta(n^{-1/2})$. Shortly after the conference version of Ref.~\cite{altug14}, Polyanskiy and Verd\'u~\cite{polyanskiy10c} relaxed some of the  conditions on the class of DMCs and also established the moderate deviations asymptotics for other important classical channels such as the additive white Gaussian noise channel. The other main contributions to the analysis of  hypothesis testing, channel coding, quantum hypothesis testing, and c-q channel coding in the different parameter regimes are summarised in Table~\ref{tb:relatedwork}. 

From a technical perspective the moderate deviations regime can be approached via a refined large deviations analysis (as was done in Ref.~\cite{altug14}) or via a variation of second-order analysis via the information spectrum method (as was proposed in Ref.~\cite{polyanskiy10c}). In our work, we mostly follow the latter approach, interspersed with ideas from large deviation theory. In particular, we build on bounds from one-shot information theory by Wang and Renner~\cite{wang10} and use techniques developed for the second-order asymptotics in Ref.~\cite{tomamicheltan14}. In concurrent work, Cheng and Hsieh~\cite{cheng17} provide a moderate deviation analysis for c-q channels via a refined error exponent analysis. Their result holds for c-q channels with finite input alphabets and their techniques are complementary to ours.

\begin{table}
	\begin{tabular}{|l|c|c|c|c|}
		\hline
		&  asymmetric binary  &  channel coding &  quantum hypothesis &  classical-quantum  \\
		& hypothesis testing &   &  testing & channel coding \\
		\hline
		large deviation ($<$) & \cite{hoeffding65} & \cite{gallager68, csiszar11} & \cite{hayashi07,nagaoka06} & unknown\footnotemark \\
		\hline
		moderate deviation ($<$) & \cite{Sas11} & \cite{altug14, polyanskiy10c} &  \em this work & \em this work  \\
		\hline
		small deviation & \cite{strassen62} & \cite{strassen62,hayashi09,polyanskiy10}  & \cite{li12,tomamichel12} & \cite{tomamicheltan14}  \\
		\hline
		moderate deviation ($>$) & \em this work & \em this work  & \em this work & \em this work \\
		\hline
		large deviation ($>$) & \cite{csiszar71,han89} & \cite{arimoto73, dueck79} & \cite{mosonyiogawa13,mosonyi14} & \cite{mosonyi14-2} \\
		\hline
	\end{tabular}
	\caption{Exposition of related work on finite resource analysis of hypothesis testing and channel coding problems. The rows correspond to different parameter regimes, labelled by the deviation from the critical rate (i.e., the relative entropy for hypothesis testing and the capacity for channel coding problems).}
	\label{tb:relatedwork}
\end{table}

\footnotetext{In contrast to classical channels a tight characterisation of the error exponent of c-q channels remains elusive to date even for high rates. See, e.g., Refs.~\cite{holevo00,hayashi07,dalai13} for partial progress.}


\section{Preliminaries}

\subsection{Notation and classical coding over quantum channels}

Let $\cH$ be a finite-dimensional Hilbert space and denote by $\cS:=\lbrace \rho\in\cH\,|\,\Tr\rho=1,\rho\geq 0\rbrace$ the quantum states on $\cH$. We take $\exp(\cdot)$ and $\log(\cdot)$ to be in an arbitrary but compatible base (such that they are inverses), and denote the natural logarithm by $\ln(\cdot)$. For convenience, we will consider the dimension of this Hilbert space to be a fixed constant, and omit any dependence constants may have on this dimension. For $\rho, \sigma \in \cS$ we write $\rho \ll \sigma$ if the support of $\rho$ is contained in the support of $\sigma$. For any closed subset $\cSo\subseteq \cS$, we will denote by $\cP(\cSo)$ the space of probability distributions supported on $\cSo$. We equip $\cS$ with the trace metric $\delta_{\Tr}(\rho,\rho'):=\frac{1}{2}\norm{\rho-\rho'}_1$ and $\cP(\cS)$ with a weak-convergence metric\footnote{An example of which is the induced L\'evy--Prokhorov metric (see, e.g., Section 6 and Theorem 6.4 in Ref.~\cite{partha67}).} $\delta_\text{wc}$, such that both are compact metric spaces with
\begin{align}
f:\cS\to \mathbb{R}\text{ continuous} \qquad\implies\qquad \mathbb{P}\mapsto \int\mathrm{d}\mathbb{P}(\rho)\,f(\rho)\text{ continuous}.
\end{align}

We will use the \emph{cumulative standard normal distribution} function $\Phi$ is defined as
\begin{align}
	\Phi(a)&:=\int_{-\infty}^{a}\frac{1}{\sqrt{2\pi}}\, e^{-\frac{x^2}2}\, \mathrm{d}x.
\end{align}

Following Ref.~\cite{tomamicheltan14}, we consider a general \emph{classical-quantum channel} $\cW: \cX \to \cS$ where $\cX$ is any set (without further structure). We define the \emph{image of the channel} as the set $\im \cW \subset \cS$ of all quantum states $\rho$ such that $\rho = \cW(x)$ for some $x \in \cX$. For convenience we assume that our Hilbert space satisfies
\begin{align}
	\cH = \mathop{\mathrm{Span}}_{\rho \in \im \cW} \mathrm{supp}(\rho)
\end{align}
such that $\sigma>0$ is equivalent to $\rho\ll\sigma$ for all $\rho\in \im \cW$.

For $M, n \in \mathbb{N}$, an \emph{$(n,M)$-code} for a classical-quantum 
channel $\cW$ is comprised of an encoder and a decoder. The \emph{encoder} is a 
map $E: \{1, 2, \ldots, M \} \to \cX^n$ and the \emph{decoder} is a positive 
operator-valued measure $\{ D_m \}_{m=1}^M$ on $\cH^{\otimes n}$. Moreover, an 
\emph{$(n,M,\eps)$-code} is an \emph{$(n,M)$-code} that satisfies
\begin{align}
	\frac{1}{M} \sum_{m=1}^M \Tr \bigg( \bigotimes_{i=1}^n \cW\bigl(E_i(m)\bigr) 
	D_m \bigg) \geq 1 - \eps \,,
\end{align}
i.e.\ the average probability of error does not exceed $\eps$.
The finite blocklength achievable region for a channel $\cW$ is the set of triples $(n,M,\eps)$ for which there exists an \emph{$(n,M,\eps)$-code} on $\cW$. We are particularly interested in the boundary
\begin{align}
	M^*(\cW;n,\eps) := \max \big\{ M \in \mathbb{N} : \exists \textrm{ a $(n,M,\eps)$-code for } \cW \big\} \,.
\end{align}
Specifically we are going to be concerned with the behaviour of the \emph{maximum rate}, which is defined as $R^*(\cW;n,\eps) := \frac{1}{n}\log M^*(\cW;n,\eps)$.

\subsection{Channel parameters}

An important parameter of a channel is the largest rate such that there exists a code of vanishing error probability in the large blocklength limit. This critical rate is known as the \emph{capacity} of a channel $C(\cW)$, which is defined as
\begin{align}
	C(\cW):=\inf_{\epsilon>0}\liminf\limits_{n\to\infty}R^*(\cW;n,\epsilon).
\end{align}
For classical-quantum channels there exists a \emph{strong converse} bound, which states that the capacity described the asymptotic rate not just for vanishing error probability, but those for non-zero fixed error probabilities as well~\cite{winter99,ogawa99}. Together with the original channel coding theorem~\cite{schumacher01,holevo98}, this yields
\begin{align}
	\lim\limits_{n\to\infty}R^*(\cW;n,\epsilon)=C(\cW)\quad \text{for all }\epsilon\in (0,1).
\end{align}

In essence the strong converse tells us that the capacity entirely dictates the asymptotic behaviour of the maximum rate at a fixed error probability. How quickly the rate approaches this asymptotic value for arbitrarily low and high error probabilities are described by the channel \emph{min-dispersion} $V_{\min}(\cW)$ and \emph{max-dispersion} $V_{\max}(\cW)$, which are defined respectively as
\begin{align}
	V_{\min}(\cW)&:=\inf_{\epsilon> 0}\limsup_{n\to\infty}\left(\frac{C(\cW)-R^*(\cW;n,\epsilon)}{\Phi^{-1}(\epsilon)/\sqrt{n}}\right)^2,\\
	V_{\max}(\cW)&:=\sup_{\epsilon<1}\limsup_{n\to\infty}\left(\frac{C(\cW)-R^*(\cW;n,\epsilon)}{\Phi^{-1}(\epsilon)/\sqrt{n}}\right)^2.
\end{align}
As with the strong converse, the min and max-dispersions also describe the dispersion at other fixed error probabilities~\cite{tomamicheltan14}:
\begin{align}
	\lim_{n\to\infty}\left(\frac{C(\cW)-R^*(\cW;n,\epsilon)}{\Phi^{-1}(\epsilon)/\sqrt{n}}\right)^2=\begin{dcases}
		V_{\min}(\cW) &\epsilon\in(0,1/2)\\
		V_{\max}(\cW) &\epsilon\in(1/2,1)
	\end{dcases}.
\end{align}

\subsection{Information quantities}

Classically, for two distributions $P$ and $Q$, the \emph{relative entropy} $D(P\|Q)$ and \emph{relative entropy variance} $V(P\|Q)$ are both defined as the mean and variance of the log-likelihood ratio $\log \bigl(P/Q\bigr)$ with respect to the distribution $P$. In the non-commutative case, for $\rho,\sigma \in \cS$ with $\rho \ll \sigma$, these definitions are generalised as~\cite{umegaki62,li12,tomamichel12}
\begin{align}
	D(\rho\|\sigma)&:=\Tr\rho\left(\log\rho-\log\sigma\right), \\
	V(\rho\|\sigma)&:=\Tr\rho\bigl(\log\rho-\log\sigma-D\left(\rho\|\sigma\right)\cdot\mathrm{id}\bigr)^2\,.
\end{align}
If $\rho \not\ll \sigma$ both quantities are set to $+\infty$. 

Following Ref.~\cite{tomamicheltan14}, for a closed set $\cSo\in \cS$, the \emph{divergence radius}\footnote{Whilst \cref{eqn:radius} characterises the divergence radius, we will mostly rely on a more useful form presented in \cref{defn:radius}.} $\chi(\cSo)$ is given by
\begin{align}
	\label{eqn:radius}
	\chi(\cSo)
	=\sup_{\mathbb{P}\in \cP(\cSo)}\int\mathrm{d}\mathbb{P}(\rho)\,D\left( \rho\, \middle\| \, \int\mathrm{d}\mathbb{P}(\rho')\,\rho' \right).
\end{align}
where $\cP(\cSo)$ denotes the space of distributions on $\cSo$. If we let $\Pi(\cSo)$ denote the distributions which achieve the above supremum, we also define the \emph{minimal and maximal peripheral variance}, $v_{\min}(\cSo)$ and $v_{\max}(\cSo)$, as
\begin{align}
	v_{\min}(\cSo):=\inf_{\mathbb{P}\in \Pi(\cSo)} \int\mathrm{d}\mathbb{P}(\rho)\,V\left( \rho\, \middle\| \, \int\mathrm{d}\mathbb{P}(\rho')\,\rho' \right),\\
	v_{\max}(\cSo):=\sup_{\mathbb{P}\in \Pi(\cSo)} \int\mathrm{d}\mathbb{P}(\rho)\,V\left( \rho\, \middle\| \, \int\mathrm{d}\mathbb{P}(\rho')\,\rho' \right).
\end{align}

For the image of a quantum channel, the above three information quantities correspond exactly to the three previously defined channel parameters~\cite{tomamicheltan14}. Specifically, for $\cSo=\overline{\im \cW}$, we have
\begin{align}
	C(\cW)=\chi(\cSo),\qquad\quad
	V_{\min}(\cW)=v_{\min}(\cSo),\qquad\quad
	V_{\max}(\cW)=v_{\max}(\cSo).
\end{align}

\subsection{Moderate deviation tail bounds}
\label{subsec:moddev}

We now discuss the relevant tail bounds we will require in the moderate deviation regime. Let $\lbrace X_{i,n}\rbrace_{i\leq n}$ be independent zero-mean random variables, and define the average variance as
\begin{align}
	V_n:=\frac{1}{n}\sum_{i=1}^{n}\Var[X_{i,n}].
\end{align}

Recall that a sequence $\lbrace t_n\rbrace_n$ is moderate if $x_n\searrow 0$ 
and $\sqrt nx_n\to+\infty$ as $n\to\infty$. Given certain bounds on the moments 
and cumulants 
of these variables, which we will make explicit below, we will see that the 
probability that the average variable $\frac{1}{n}\sum_{i=1}^{n}X_{i,n}$ 
deviates from the mean by a moderate sequence $\lbrace t_n\rbrace_n$ decays 
asymptotically as
\begin{align}
	\label{eqn:asymp}
	\ln\Pr\left[\frac{1}{n}\sum_{i=1}^{n}X_{i,n}\geq t_n\right]=-\bigl(1+o(1)\bigr)\frac{nt_n^2}{2V_n}.
\end{align}

\begin{lem}[Moderate deviation lower bound]
	\label{lem:moddev lower}
	If there exist constants $\nu>0$ and $\tau$ such that $\nu\leq V_n$ and 
	\begin{align}
	\frac{1}{n}\sum_{i=1}^{n}\mathbb{E}\left[\abs{X_{i,n}}^3\right]\leq \tau
	\end{align} 
	for all $n$, then for any $\eta>0$ there exists a constant $N(\lbrace t_i\rbrace,\nu,\tau,\eta)
	$ such that, for all $n\geq N$, the probability of a moderate deviation is lower bounded as
	\begin{align}
	\ln\Pr\left[\frac{1}{n}\sum_{i=1}^nX_{i,n}\geq t_n\right]\geq 
	-(1+\eta)\frac{nt_n^2}{2V_n}.
	\end{align}
\end{lem}

\begin{lem}[Moderate deviation upper bound]
	\label{lem:moddev upper}
	If there exists a constant $\gamma$ such that 
	\begin{align}
	\frac{1}{n}\sum_{i=1}^{n}\sup_{s\in[0,1/2]}\abs{\frac{\mathrm{d}^3}{\mathrm{d}s^3}\ln\mathbb{E}\left[e^{sX_{i,n}}\right]}\leq \gamma,
	\end{align}
	for all $n$, then for any $\eta>0$ there exists a constant $N(\lbrace t_i\rbrace, \gamma,\eta)
	$ such that, for all $n\geq N$, the probability of a moderate deviation is upper bounded as
	\begin{align}
	\ln \Pr\left[\frac{1}{n}\sum_{i=1}^{n}X_{i,n}\geq t_n\right]\leq -\frac{nt_n^2}{2V_n+\eta}.
	\end{align}
\end{lem}

If $V_n$ has a uniform lower bound, then as $\eta\searrow 0$ the above two bounds sandwich together, giving the two-sided asymptotic scaling of Eq.~\ref{eqn:asymp}. In this case we can see that
\begin{align}
	\sigma\left[\frac{1}{n}\sum_{i=1}^{n}X_i\right]=\sqrt{V_n/n}=\Theta(1/\sqrt{n})
	\qquad\text{and}\qquad
	\sqrt{\frac{1}{n}\sum_{i=1}^{n}\sigma^2\left[X_i\right]}=\sqrt{V_n}=\Theta(1),
\end{align}
where $\sigma\left[\cdot\right]$ denotes the standard deviation. If we interpret the standard deviation as setting the `length-scale' on which a distribution decays, then the above two quantities---the deviation of the average, and average\footnote{More specifically the root-mean-square} of the deviation---set the length-scales of small and large deviation bounds respectively. Using this intuition, we can generalise moderate deviation bounds to give tight two-sided bounds for distributions with arbitrary normalisation, in which $V_n$ is no longer bounded. To do this we will tail bound for deviations which are moderate, \emph{in units of }$\sqrt{V_n}$.

\begin{cor}[Dimensionless moderate deviation bound]
	\label{cor:moddev nondim}
	If there exists a $\gamma$ such that
	\begin{align}
		\frac{1}{nV_n^{3/2}}\sum_{i=1}^{n}\sup_{s\in[0,1/2]}\abs{\frac{\mathrm{d}^3}{\mathrm{d}s^3}\ln\mathbb{E}\left[e^{sX_{i,n}}\right]}\leq \gamma,
	\end{align}
	for all $n$,
	then there exists a constant $N(\lbrace t_i\rbrace,\gamma)$ such that, for all $n\geq N$, we have the two-sided bound 
	\begin{align}
		-(1+\eta)\frac{nt_n^2}{2}\leq \ln \Pr\left[\frac{1}{n}\sum_{i=1}^{n}X_{i,n}\geq t_n\sqrt{V_n}\right]\leq -(1-\eta)\frac{nt_n^2}{2}.
	\end{align}
\end{cor}

We present proofs of these lemmas in \cref{app:tailbounds}.

\subsection{Reversing lemma}

Intuitively one might expect that moderate deviation bounds can be `reversed' e.g.\ that the bound on the probability given the deviation (see Lemmas \ref{lem:moddev lower} and \ref{lem:moddev upper}) of the form
\begin{align}
	\lim\limits_{n\to \infty}\frac{V_n}{nt_n^2}\ln\Pr\left[\frac{1}{n}\sum_{i=1}^{n}X_i\geq t_n\right]=-\frac{1}{2},
\end{align}
is equivalent to a bound on the deviation given the probability
\begin{align}
	\lim\limits_{n\to\infty}\frac{1}{t_n}\inf\left\lbrace t\in\mathbb{R} \,\middle|\, \frac{V_n}{nt_n^2}\ln\Pr\left[\frac{1}{n}\sum_{i=1}^{n}X_i\geq t\right]\leq -\frac{1}{2} \right\rbrace=
\end{align}

We will now see that such an ability to `reverse' moderate deviation bounds is generic. We do this by considering two quantities $A$ and $B$ defined on the same domain, and considering the infimum value of each quantity for a fixed value of the other. 

\begin{lem}[Reversing Lemma]
	\label{lem:reverse}
	Let $\lbrace A_i\rbrace_i$ and $\lbrace B_i\rbrace_i$ be sequences of real functions with $\inf_{t} A_i(t)\leq 0$ and $\inf_{t}B_i(t)\leq 0$ for all $i$. If we define $\hat{A}_n(b):=\inf_{t} \left\lbrace A_n(t) \middle| B_n(t)\leq b \right\rbrace$ and $\hat{B}_n(a):=\inf_{t} \left\lbrace B_n(t) \middle| A_n(t)\leq a \right\rbrace$, then
	\begin{align}
		\lim\limits_{n\to\infty}\frac{\hat{A}_n(b_n)}{b_n}=1, \quad\forall \lbrace b_n\rbrace\text{ moderate}
		~~\qquad&\Longleftrightarrow\qquad~~ 
		\lim\limits_{n\to\infty}\frac{\hat{B}_n(a_n)}{a_n}=1, \quad\forall \lbrace a_n\rbrace\text{ moderate}.
	\end{align}
\end{lem}
\begin{proof}
	See \cref{app:reverse}.
\end{proof}


\section{Hypothesis testing}
\label{sec:hypo}

Whilst the divergence radius characterises the channel capacity, one-shot channel bounds are characterised by a quantity known as the \emph{$\epsilon$-hypothesis testing divergence}~\cite{wang10}. As the name suggests, as well as being relevant to one-shot channel coding bounds, the hypothesis testing divergence also has an operational interpretation in the context of hypothesis testing of quantum states. We will start by considering a moderate deviation analysis of this quantity.

\subsection{Hypothesis testing divergence}

Consider a hypothesis testing problem, in which $\rho$ and $\sigma$ correspond to the null and alternative hypotheses respectively. A test between these hypotheses will take the form of a POVM $\lbrace Q,I-Q\rbrace$, where $0\leq Q\leq I$. For a given $Q$, the type-I and type-II error probabilities are given by 
\begin{align}
	\alpha(Q;\rho,\sigma):=\Tr (I-Q)\rho,\qquad\qquad
	\beta(Q;\rho,\sigma):=\Tr Q\sigma.
\end{align}
If we define the smallest possible type-II error given a type-I error at most $\epsilon$ as
\begin{align}
	\beta_\epsilon(\rho\|\sigma):=\min_{0\leq Q\leq \mathbb{I}}\left\lbrace \beta(Q;\rho,\sigma) \,\middle|\, \alpha(Q;\rho,\sigma)\leq\epsilon \right\rbrace,
\end{align}
then the $\epsilon$-hypothesis testing divergence is defined as
\begin{align}
	D^\epsilon_{\mathrm{h}}(\rho\|\sigma):=-\log\frac{\beta_{\epsilon}(\rho\|\sigma)}{1-\epsilon}.
\end{align}
We note that the denominator of $1-\epsilon$ follows the normalisation in~\cite{dupuis12} such that $D_{\mathrm{h}}^\epsilon(\rho\|\rho)=0$ for all $\rho$.

An obvious extension of this hypothesis problem is to the case of $n$ copies of each state, i.e.\ a hypothesis test between $\rho^{\otimes n}$ and $\sigma^{\otimes n}$, or more generally between two product states $\otimes_{i=1}^n\rho_i$ and $\otimes_{i=1}^n\sigma_i$. A second-order analysis of the $\epsilon$-hypothesis testing divergence for a non-vanishing $\eps$ was given in~\cite{li12,tomamichel12}.

\begin{thm}[Moderate deviation of the hypothesis testing divergence]
	For any moderate sequence $\lbrace a_n\rbrace_n$ and states $\lbrace \rho_n\rbrace_n$ and $\lbrace \sigma_n\rbrace_n$ such that both $\lambda_{\min}(\sigma_i)$ and $V(\rho_i\|\sigma_i)$ are both uniformly bounded away from zero, the $\epsilon_n$- and $(1-\epsilon_n)$-hypothesis testing divergences of non-uniform product states for $\epsilon_n=e^{-na_n^2}$ scale as
	\begin{align}
		\frac{1}{n}D_{\mathrm{h}}^{\epsilon_n}\left( \bigotimes_{i=1}^n\rho_i \middle\|\,\bigotimes_{i=1}^n\sigma_i\right)&=D_n-\sqrt{2V_n}\,a_n+o(a_n), \\
		\frac{1}{n}D_{\mathrm{h}}^{1-\epsilon_n}\left( \bigotimes_{i=1}^n\rho_i \middle\|\,\bigotimes_{i=1}^n\sigma_i\right)&=D_n+\sqrt{2V_n}\,a_n+o(a_n),
	\end{align}
	where $D_n:=\frac{1}{n}\sum_{i=1}^{n}D(\rho_i\|\sigma_i)$ and $V_n:=\frac{1}{n}\sum_{i=1}^{n}V(\rho_i\|\sigma_i)$. More specifically for any $\rho$ and $\sigma$ such that $\rho\ll \sigma$, the hypothesis testing divergences of uniform product states scale as
	\begin{align}
		\frac{1}{n}D_{\mathrm{h}}^{\epsilon_n}\left( \rho^{\otimes n} \middle\|\sigma^{\otimes n}\right)&=D(\rho\|\sigma)-\sqrt{2V(\rho\|\sigma)}\,a_n+o(a_n), \label{eq:hypo-mod}\\
		\frac{1}{n}D_{\mathrm{h}}^{1-\epsilon_n}\left( \rho^{\otimes n} \middle\|\sigma^{\otimes n}\right)&=D(\rho\|\sigma)+\sqrt{2V(\rho\|\sigma)}\,a_n+o(a_n).
	\end{align}
\end{thm}

In Sect.~\ref{subsec:inward} we will bound the regularised hypothesis testing divergences towards the relative entropy (the \emph{inward bound}), and in Sect.~\ref{subsec:outward} we will bound them away (the \emph{outward bound}).

\begin{remark}
	For sequences $\eps_n$ bounded away from zero and one the second-order expansion in Refs.~\cite{li12,tomamichel12} yields 
	\begin{align}
		\frac{1}{n}D_{\mathrm{h}}^{\epsilon_n}\left( \rho^{\otimes n} \middle\|\sigma^{\otimes n}\right) = D(\rho\|\sigma) + \sqrt{\frac{V(\rho\|\sigma)}{n}}\, \Phi^{-1}(\eps_n)  + O\left(\frac{\log n}{n}\right) , \label{eq:hypo-so}
	\end{align}
	where $\Phi$ denotes the cumulative distribution function of the standard normal.
	As already pointed out in Ref.~\cite{polyanskiy10c}, for small $\eps_n$ we have $\Phi^{-1}(\eps_n) \approx \sqrt{- 2 \ln \eps_n}$. Ignoring all higher order terms, the substitution $\eps_n = e^{-n a_n^2}$ into~\eqref{eq:hypo-so} then recovers the expression in~\eqref{eq:hypo-mod}. In this sense the two results thus agree at the boundary between small and moderate deviations.
\end{remark}

\begin{remark}
	A similar argument can be sketched at the boundary between moderate and large 
	deviations.
	The quantum Hoeffding bound~\cite{nagaoka06,hayashi07} states that if $\frac{1}{n}D_{\mathrm{h}}^{\epsilon_n}(\rho^{\otimes n}\|\sigma^{\otimes n})\leq D(\rho\|\sigma) - r$ for some small $r > 0$ then $\eps_n$ drops exponentially in $n$ with the exponent given by
	\begin{align}
		\sup_{0\leq\alpha<1}\frac{\alpha-1}{\alpha}\Bigl[D(\rho\|\sigma)-r-D_{\alpha}(\rho\|\sigma)\Bigr] , \label{eqn:expr1}
	\end{align}
	where $D_{\alpha}(\rho\|\sigma)$ is the Petz' quantum R\'enyi relative entropy~\cite{petz86}. For sufficiently small $r$, the expression in~\eqref{eqn:expr1} attains its supremum close to $\alpha = 1$ and we can thus approximate $D_{\alpha}(\rho\|\sigma) \approx D(\rho\|\sigma) + \frac{\alpha-1}{2} V(\rho\|\sigma)$ by its Taylor expansion~\cite{lintomamichel14}. Evaluating this approximate expression yields 
	\begin{align}
		\eps_n = e^{- n\frac{r^2}{2V(\rho\|\sigma)}} \,. \label{eqn:expr2}
	\end{align} 
	up to leading order in $r$. Substituting $r = \sqrt{2V(\rho\|\sigma)}a_n$ into~\eqref{eqn:expr2} then recovers~\eqref{eq:hypo-mod}.
	An essentially equivalent argument is also applicable to the strong converse exponent derived in Ref.~\cite{mosonyiogawa13}.
\end{remark}

\subsection{Nussbaum--Szko\l a distributions}

To allow us to apply a moderate deviation analysis to the quantum hypothesis testing divergence, we leverage the results of Ref.~\cite{tomamichel12} which allow us to reduce the hypothesis testing divergence of quantum states to a quantity known as the information spectrum divergence of certain classical distributions, known as the Nussbaum--Szko\l a distributions. 

\begin{defn}[Nussbaum--Szko\l a distributions~\cite{NussbaumSzkola2009}]
	The \emph{Nussbaum--Szko\l a distributions} for a pair of states $\rho$ and $\sigma$ are given by
	\begin{align}
		P^{\rho,\sigma}(a,b)=r_a\abs{\braket{\phi_a}{\psi_b}}^2\qquad\text{and}\qquad Q^{\rho,\sigma}(a,b)=s_b\abs{\braket{\phi_a}{\psi_b}}^2
	\end{align}
	where the states are eigendecomposed as $\rho=\sum_ar_a\ketbra{\phi_a}{\phi_a}$ and $\sigma=\sum_bs_b\ketbra{\psi_b}{\psi_b}$. 
\end{defn}

The power of the Nussbaum--Szko\l a distributions lies in their ability to reproduce both the divergence and variance of the underlying quantum states
\begin{align}
	D(\rho\|\sigma)=D(P^{\rho,\sigma}\|Q^{\rho,\sigma})
	,
	\qquad \text{and}\qquad 
	V(\rho\|\sigma)=V(P^{\rho,\sigma}\|Q^{\rho,\sigma})
	.
\end{align}
As well as capturing these asymptotic quantities, the hypothesis testing relative entropy, which arises one-shot channel coding  bounds, can also be captured by the Nussbaum--Szko\l a distributions. Specifically this is done via the \emph{information spectrum divergence}, which is defined for two classical distributions $P$ and $Q$ by a tail bound on the log-likelihood ratio as
\begin{align}
	D_{\mathrm{s}}^{\epsilon}(P\|Q):=\sup\left\lbrace R ~\middle|~ \Pr_{X\leftarrow P}\left[\log \frac{P(X)}{Q(X)}\leq R\right]\leq \epsilon \right\rbrace.
\end{align}
Inserting the Nussbaum--Skzo\l a distributions, we find that the (classical) information spectrum divergence approximates the (quantum) hypothesis testing divergence.
\begin{lem}[Thm.~14, Ref.~\cite{tomamichel12}]
	\label{lem:infospecdiv}
	There exists a universal constant $K$ such that for any states $\rho$ and $\sigma$ with $\lambda_{\min}(\sigma)\geq \lambda$ and $\epsilon< 1/2$, we find that $D_{\mathrm{h}}^\epsilon(\rho\|\sigma)$ is bounded as
	\begin{align}
		D_{\mathrm{h}}^\epsilon(\rho\|\sigma)&\leq D_{\mathrm{s}}^{2\epsilon}(P^{\rho,\sigma}\|Q^{\rho,\sigma})+\log\frac{1-\epsilon}{\epsilon^3(1-2\epsilon)}+\log K \lceil \ln(1/\lambda)\rceil\\
		D_{\mathrm{h}}^\epsilon(\rho\|\sigma)&\geq D_{\mathrm{s}}^{\epsilon/2}(P^{\rho,\sigma}\|Q^{\rho,\sigma})-\log\frac{1}{\epsilon(1-\epsilon)}-\log K \lceil \ln(1/\lambda)\rceil,
	\end{align}
	and $D_{\mathrm{h}}^{1-\epsilon}(\rho\|\sigma)$ is bounded as
	\begin{align}
		D_{\mathrm{h}}^{1-\epsilon}(\rho\|\sigma)&\leq D_{\mathrm{s}}^{1-\epsilon/2}(P^{\rho,\sigma}\|Q^{\rho,\sigma})+\log\frac{1-\epsilon/2}{\epsilon^4}+\log K \lceil \ln(1/\lambda)\rceil\\
		D_{\mathrm{h}}^{1-\epsilon}(\rho\|\sigma)&\geq D_{\mathrm{s}}^{1-2\epsilon}(P^{\rho,\sigma}\|Q^{\rho,\sigma})-\log\frac{1}{\epsilon^2}-\log K \lceil \ln(1/\lambda)\rceil.
	\end{align}
\end{lem}

As the information spectrum divergence is defined in terms of a tail bound, we will bound these quantities using the moderate deviation tail bounds of Sect.~\ref{subsec:moddev}. To do this, we will start by showing that the log-likelihood ratio of Nussbaum--Skzo\l a distributions is sufficiently well behaved, specifically that its cumulant generating function has bounded derivatives.

\begin{lem}[Bounded cumulants]
	\label{lem:boundedcumulants}
	For $\lambda>0$, there exists constants $C_k(\lambda)$ such that the cumulant generating function $h(t):=\ln \mathbb{E}\left[e^{tZ}\right]$ of the log-likelihood ratio $Z:=\log P^{\rho,\sigma}/Q^{\rho,\sigma}$ for $\lambda_{\min}(\sigma)\geq \lambda$ is smooth and has uniformly bounded derivatives in a neighbourhood of the origin
	\begin{align}
	\sup_{\abs{t}\leq 1/2}\abs{\frac{\partial^k}{\partial t^k}h(t)} \leq C_k.
	\end{align}
\end{lem}
We present a proof of this lemma in \cref{app:cumulant}.

\subsection{Inward bound}
\label{subsec:inward}
\begin{prop}[Inward bound]
	\label{prop:HTD-inward}
	For any constants $\lambda,\eta>0$, there exists a constant $N(\lbrace a_i\rbrace,\lambda,\eta)$ such that, for $n\geq N$, the hypothesis testing divergence can be bounded for any states $\lbrace\rho_i\rbrace_i$ and $\lbrace \sigma_i\rbrace_i$ with
	$\lambda_{\min}(\sigma_i)\geq\lambda$ as
	\begin{align}
		\frac{1}{n}D_{\mathrm{h}}^{\epsilon_n}\left( \bigotimes_{i=1}^n\rho_i \middle\|\,\bigotimes_{i=1}^n\sigma_i\right)&\geq D_n-\sqrt{2V_n}a_n-\eta a_n,\\
		\frac{1}{n}D_{\mathrm{h}}^{1-\epsilon_n}\left( \bigotimes_{i=1}^n\rho_i \middle\|\,\bigotimes_{i=1}^n\sigma_i\right)&\leq D_n+\sqrt{2V_n}\,a_n+\eta a_n.
	\end{align}
	where $D_n:=\frac{1}{n}\sum_{i=1}^{n}D(\rho_i\|\sigma_i)$ and $V_n:=\frac{1}{n}\sum_{i=1}^{n}V(\rho_i\|\sigma_i)$.
\end{prop}
\begin{proof}
	Firstly, let $Z_i$ be the log-likelihood ratios 
	\begin{align}
	Z_i:=\log\frac{P^{\rho_i,\sigma_i}(A_i,B_i)}{Q^{\rho_i,\sigma_i}(A_i,B_i)}, \qquad (A_i,B_i)\leftarrow P^{\rho_i,\sigma_i}.
	\end{align}
	In terms of these log-likelihood ratios, the lower and upper bound on the $\epsilon_n$- and $(1-\epsilon_n)$-hypothesis testing divergences respectively from \cref{lem:infospecdiv} become
	\begin{align}
		D_{\mathrm{h}}^{\epsilon_n}\left( \bigotimes_{i=1}^n\rho_i \middle\|\,\bigotimes_{i=1}^n\sigma_i\right)&\geq \sup\left\lbrace R~\middle|~ \Pr\Biggl[\sum_{i=1}^n Z_i\leq R\Biggr]\leq \epsilon_n/2 \right\rbrace-\log \frac{1}{\epsilon_n(1-\epsilon_n)}-\log Kn\lceil\ln(1/\lambda) \rceil,\\
		D_{\mathrm{h}}^{1-\epsilon_n}\left( \bigotimes_{i=1}^n\rho_i \middle\|\,\bigotimes_{i=1}^n\sigma_i\right)&\leq \sup\left\lbrace R~\middle|~ \Pr\Biggl[\sum_{i=1}^n Z_i\leq R\Biggr]\leq 1-\epsilon_n/2 \right\rbrace+\log \frac{1-\epsilon_n/2}{\epsilon_n^4}+\log Kn\lceil\ln(1/\lambda) \rceil.
	\end{align}
	
	Recalling that $\epsilon_n:=e^{-na_n^2}$, we can see that in both cases the error terms scale like $\Theta(na_n^2)$ and $\Theta(\log n)$ respectively, which are both $o(na_n)$. As such, there must exist an $N_1(\lbrace a_i\rbrace,\lambda,\eta)$ such that, for $n\geq N_1$, these error terms are bounded by $\eta na_n/2$ as
	\begin{align}
		\frac{1}{n}D_{\mathrm{h}}^{\epsilon_n}\left( \bigotimes_{i=1}^n\rho_i \middle\|\,\bigotimes_{i=1}^n\sigma_i\right)&\geq \frac{1}{n}\sup\left\lbrace R~\middle|~ \Pr\Biggl[\sum_{i=1}^n Z_i\leq R\Biggr]\leq \epsilon_n/2 \right\rbrace-\eta a_n/2,\\
		\frac{1}{n}D_{\mathrm{h}}^{1-\epsilon_n}\left( \bigotimes_{i=1}^n\rho_i \middle\|\,\bigotimes_{i=1}^n\sigma_i\right)&\leq \frac{1}{n}\sup\left\lbrace R~\middle|~ \Pr\Biggl[\sum_{i=1}^n Z_i\leq R\Biggr]\leq 1-\epsilon_n/2 \right\rbrace+\eta a_n /2.
	\end{align}
	
	Next we want to apply the tail bounds of Sect.~\ref{subsec:moddev}. To this end, we will start by defining zero-mean variables $X_i:=Z_i-D(\rho_i\|\sigma_i)$. In terms of these variables, the above bounds take the form
	\begin{align}
		\frac{1}{n}D_{\mathrm{h}}^{\epsilon_n}\left( \bigotimes_{i=1}^n\rho_i \middle\|\,\bigotimes_{i=1}^n\sigma_i\right)&\geq D_n-\inf\left\lbrace t~\middle|~ \Pr\Biggl[\frac{1}{n}\sum_{i=1}^n (-X_i)\geq t\Biggr]\leq \epsilon_n/2 \right\rbrace-\eta a_n/2,\\
		\frac{1}{n}D_{\mathrm{h}}^{1-\epsilon_n}\left( \bigotimes_{i=1}^n\rho_i \middle\|\,\bigotimes_{i=1}^n\sigma_i\right)&\leq D_n+\inf\left\lbrace t~\middle|~ \Pr\Biggl[\frac{1}{n}\sum_{i=1}^n \left(+X_i\right)\geq  t\Biggr]\leq \epsilon_n/2 \right\rbrace+\eta a_n/2.
	\end{align}
	
	By \cref{lem:boundedcumulants} there exists constants $\bar V(\lambda)$ and $\gamma(\lambda)$, such that $V_i\leq \bar V$ and
	\begin{align}
	\sup_{t\in[0,1/2]}\abs{\frac{\mathrm{d}^3}{\mathrm{d}s^3}\ln\mathbb{E}\left[e^{s(\pm X_i)}\right]}\leq \gamma
	\end{align}
	for all $i$. If we let $t_n:=\left(\sqrt{2V_n}+\eta/2\right)a_n$, then \cref{lem:moddev upper} gives an $N_2(\lbrace a_i\rbrace,\lambda,\eta)$ such that, for $n\geq N_2$, the tail probability is bounded as
	\begin{align}
		\ln\Pr\Biggl[\frac{1}{n}\sum_{i=1}^n (\pm X_i)\geq t_n\Biggr]
		&\leq \frac{-nt_n^2}{2V_n+\eta^2 /3}\\
		&\leq -\frac{\left(\sqrt{2V_n}+\eta/2\right)^2}{2V_n+\eta^2/5}na_n^2\\
		&\leq -\frac{2V_n+\eta^2/4}{2V_n+\eta^2/5}na_n^2\\
		&= -\left(1+\frac{\eta^2}{40V_n+4\eta^2}\right)na_n^2\\
		&\leq -\left(1+\frac{\eta^2}{40\bar{V}+4\eta^2}\right)na_n^2.
	\end{align}
	As $\eta^2/(40\bar{V}+4\eta)$ is a constant and $na_n^2\to \infty$, there must exist a constant $N_3(\lbrace a_i\rbrace,\lambda,\eta)$ such $n\geq N_3$ implies
	\begin{align}
		\ln\Pr\Biggl[\frac{1}{n}\sum_{i=1}^n (\pm X_i)\geq t_n\Biggr]\leq -na_n^2-1=\ln (\epsilon_n/2),
	\end{align}
	and therefore that 
	\begin{align}
		\inf\left\lbrace t~\middle|~ \Pr\Biggl[\frac{1}{n}\sum_{i=1}^n (\pm X_i)\geq t\Biggr]\leq \epsilon_n/2 \right\rbrace\leq t_n.
	\end{align}
	Putting everything together, we get that for any $n\geq N(\lbrace a_i\rbrace,\lambda,\eta):=\max\lbrace N_1,N_2,N_3\rbrace$ we have
	\begin{align}
		\frac{1}{n}D_{\mathrm{h}}^{\epsilon_n}\left( \bigotimes_{i=1}^n\rho_i \middle\|\,\bigotimes_{i=1}^n\sigma_i\right)&\geq D_n-\sqrt{2V_n}\,a_n-\eta a_n,\\
		\frac{1}{n}D_{\mathrm{h}}^{1-\epsilon_n}\left( \bigotimes_{i=1}^n\rho_i \middle\|\,\bigotimes_{i=1}^n\sigma_i\right)&\leq D_n+\sqrt{2V_n}\,a_n+\eta a_n.
	\end{align}
	as required.
\end{proof}

\subsection{Outward bound}
\label{subsec:outward}

\begin{prop}[Outward bound]
	\label{prop:HTD-outward}
	For any constants $\lambda,\eta>0$, there exists a constant $N(\lbrace a_i\rbrace,\lambda,\eta)$ such that, for $n\geq N$, the hypothesis testing divergence can be bounded for any states $\lbrace\rho_i\rbrace_i$ and $\lbrace\sigma_i\rbrace_i$ with
	$\lambda_{\min}(\sigma_i)\geq\lambda$ as
	\begin{align}
		\frac{1}{n}D_{\mathrm{h}}^{\epsilon_n}\left( \bigotimes_{i=1}^n\rho_i \middle\|\,\bigotimes_{i=1}^n\sigma_i\right)&\leq D_n+\eta a_n,\\
		\frac{1}{n}D_{\mathrm{h}}^{1-\epsilon_n}\left( \bigotimes_{i=1}^n\rho_i \middle\|\,\bigotimes_{i=1}^n\sigma_i\right)&\geq D_n-\eta\,a_n.
	\end{align}
	where $D_n:=\frac{1}{n}\sum_{i=1}^{n}D(\rho_i\|\sigma_i)$. Moreover, if we 
	let $V_n:=\frac{1}{n}\sum_{i=1}^{n}V(\rho_i\|\sigma_i)$, and there also 
	exists a constant $\nu>0$ such that $V_i\geq \nu$ for all $i$, then there 
	exists an $N'(\lbrace a_i\rbrace,\lambda,\nu,\eta)$ such that, for $n\geq 
	N'$, the hypothesis testing divergence is more tightly bounded as
	\begin{align}
		\frac{1}{n}D_{\mathrm{h}}^{\epsilon_n}\left( \bigotimes_{i=1}^n\rho_i \middle\|\,\bigotimes_{i=1}^n\sigma_i\right)&\leq D_n-\sqrt{2V_n}a_n+\eta a_n,\\
		\frac{1}{n}D_{\mathrm{h}}^{1-\epsilon_n}\left( \bigotimes_{i=1}^n\rho_i \middle\|\,\bigotimes_{i=1}^n\sigma_i\right)&\geq D_n+\sqrt{2V_n}\,a_n-\eta a_n.
	\end{align}
\end{prop}
\begin{proof}
	Similar to \cref{prop:HTD-inward}, we will start by taking the upper and lower bounds on the $\epsilon_n$- and $(1-\epsilon_n)$-hypothesis testing divergences respectively from \cref{lem:infospecdiv}. 
	This gives that there exists an $N_1(\lbrace a_i\rbrace,\lambda,\eta)$ such that, for $n\geq N_1$, we have
	\begin{align}
		\frac{1}{n}D_{\mathrm{h}}^{\epsilon_n}\left( \bigotimes_{i=1}^n\rho_i \middle\|\,\bigotimes_{i=1}^n\sigma_i\right)&\leq D_n-\inf\left\lbrace t~\middle|~ \Pr\Biggl[\frac{1}{n}\sum_{i=1}^n (-X_i)\geq t\Biggr]\leq 2\epsilon_n \right\rbrace+\eta a_n/2,\\
		\frac{1}{n}D_{\mathrm{h}}^{1-\epsilon_n}\left( \bigotimes_{i=1}^n\rho_i \middle\|\,\bigotimes_{i=1}^n\sigma_i\right)&\geq D_n+\inf\left\lbrace t~\middle|~ \Pr\Biggl[\frac{1}{n}\sum_{i=1}^n \left(+X_i\right)\geq  t\Biggr]\leq 2\epsilon_n \right\rbrace-\eta a_n/2.
	\end{align}
	where $X_i:=Z_i-D(\rho_i\|\sigma_i)$.
	
	Firstly, applying Chebyshev's inequality two standard deviations below the mean gives us that
	\begin{align}
		\Pr\Biggl[\frac{1}{n}\sum_{i=1}^n (\pm X_i)\geq  -2\sqrt{V_n/n}\Biggr]\geq 3/4\geq2\epsilon_n,
	\end{align}
	and so we conclude that
	\begin{align}
		\inf\left\lbrace t~\middle|~ \Pr\Biggl[\frac{1}{n}\sum_{i=1}^n \left(\pm X_i\right)\geq  t\Biggr]\leq 2\epsilon_n \right\rbrace\geq -2\sqrt{V_n/n}.
	\end{align}
	By \cref{lem:boundedcumulants}, $V_n$ must be bounded $V_n\leq \bar{V}(\lambda)$, and thus $\sqrt{V_n/n}=\mathcal{O}(1/\sqrt{n})=o(a_n)$. As such, there must exist an $N_2(\lbrace a_i\rbrace,\lambda,\eta)$ such that $n\geq N_2$ implies $2\sqrt{V_n/n}\leq \eta a_n/2$. Inserting this tail bound, we get that for any $n\geq N(\lbrace a_i\rbrace,\lambda,\eta):=\max\lbrace N_1,N_2\rbrace$ that
	\begin{align}
		\frac{1}{n}D_{\mathrm{h}}^{\epsilon_n}\left( \bigotimes_{i=1}^n\rho_i \middle\|\,\bigotimes_{i=1}^n\sigma_i\right)&\leq D_n-\eta a_n,\\
		\frac{1}{n}D_{\mathrm{h}}^{1-\epsilon_n}\left( \bigotimes_{i=1}^n\rho_i \middle\|\,\bigotimes_{i=1}^n\sigma_i\right)&\geq D_n+\eta a_n,
	\end{align}
	as required.
	
	If there also exists an $\nu > 0$ such that $V_i\geq \nu$, then we can use a more refined moderate deviation bound. Specifically, \cref{lem:boundedcumulants} gives us a bound on the absolute third moment of $X_i$, which allows us to apply \cref{lem:moddev lower}. If we let $t_n:=(\sqrt{2V_n}-\eta/2)a_n$ and assume $\eta<\sqrt{8\nu}$ such that $\lbrace t_n\rbrace_n$ is moderate, then this gives us that there exists an $N_3(\lbrace a_i\rbrace, \lambda,\nu,\eta)$ such that, for any $n\geq N_3$, the tail probabilities are bounded
	\begin{align}
		\ln \Pr\Biggl[\frac{1}{n}\sum_{i=1}^n (\pm X_i)
		\geq t_n\Biggr]
		&\geq -\Bigl(1+\eta/\sqrt{2\bar{V}}\,\Bigr)\frac{nt_n^2}{2V_n} \\
		&\geq -\frac{\left(1+\eta/\sqrt{2\bar{V}}\right)\left(\sqrt{2V_n}-\eta/2\right)^2}{2V_n}na_n^2 \\
		&\geq -\frac{\left(1+\eta/\sqrt{2V_n}\right)\left(\sqrt{2V_n}-\eta/2\right)^2}{2V_n}na_n^2 \\
		&\geq -\left(1-\frac{5\eta^2}{8\bar V}\right)na_n^2.
	\end{align}
	Once again, the second term in the parenthesis is a non-zero constant, and thus there must exist an $N_4(\lbrace a_i\rbrace,\lambda,\eta)$ such that
	\begin{align}
		\log \Pr\Biggl[\frac{1}{n}\sum_{i=1}^n (\pm X_i)
		\geq t_n\Biggr]\geq -na_n^2+1=\ln 2\epsilon_n,
	\end{align}
	allowing us to conclude $\Pr\left[\frac{1}{n}\sum_{i=1}^n (\pm X_i)
	\geq t_n\right]\geq 2\epsilon_n$, and therefore 
	\begin{align}
		\inf\left\lbrace t~\middle|~ \Pr\Biggl[\frac{1}{n}\sum_{i=1}^n X_i\geq t\Biggr]\leq 2\epsilon_n \right\rbrace\geq t_n.
	\end{align}
	Inserting this into the above bounds, we find that for any $n\geq N'(\lbrace a_i\rbrace,\lambda,\nu,\eta):=\max\lbrace N_1,N_3,N_4\rbrace$, we have the desired final bound
	\begin{align}
		\frac{1}{n}D_{\mathrm{h}}^{\epsilon_n}\left( \bigotimes_{i=1}^n\rho_i \middle\|\,\bigotimes_{i=1}^n\sigma_i\right)&\leq D_n-\sqrt{2V_n}\,a_n+\eta a_n,\\
		\frac{1}{n}D_{\mathrm{h}}^{1-\epsilon_n}\left( \bigotimes_{i=1}^n\rho_i \middle\|\,\bigotimes_{i=1}^n\sigma_i\right)&\geq D_n+\sqrt{2V_n}\,a_n-\eta a_n.
	\end{align}
\end{proof}


\section{Channel Coding}
\label{sec:channels}

We are now going to show how the above moderate deviation bounds can be applied to the capacity of a classical-quantum channel. 

\begin{thm}[Moderate deviation of c-q channels]
	\label{thm:moddevcode}
	For any moderate sequence $\lbrace a_n\rbrace_n$ and memoryless c-q channel 
	$\cW$ with capacity $C(\cW)$ and min-dispersion $V_{\min}(\cW)$, being 
	operated at error probability no larger than $\epsilon_n:=e^{-na_n^2}$, the 
	optimal rate deviates below the capacity as
	\begin{align}
		R^*\left(\cW ;n,\epsilon_n\right)=C(\cW)-\sqrt{2V_{\min}(\cW)}\,a_n+o(a_n).
	\end{align}
	Conversely, if the channel has max-dispersion $V_{\max}$ and is operated at 
	error probability no larger than $1-\epsilon_n$, then the optimal rate 
	deviates above the capacity as
	\begin{align}
		R^*\left(\cW;n,1-\epsilon_n\right)=C(\cW)+\sqrt{2V_{\max}(\cW)}\,a_n+o(a_n).
	\end{align}
\end{thm}
If either the min- or max-dispersion is non-zero, an application of \cref{lem:reverse} gives an equivalent formulation in terms of the minimal error probability at a given rate.

\begin{cor}
	For any moderate sequence $\lbrace s_n\rbrace_n$, the error probability for a code with min-dispersion $V_{\min}>0$ deviating below capacity by $s_n$ scales as
	\begin{align}
		\lim\limits_{n\to \infty}\frac{1}{ns_n^2}\ln\epsilon^*(\cW;n,C-s_n)=-\frac{1}{2V_{\min}}.
	\end{align}
	Similarly, for a code with max-dispersion $V_{\max}>0$ deviating above capacity by $s_n$, the error probability scales
	\begin{align}
	\lim\limits_{n\to \infty}\frac{1}{ns_n^2}\ln\bigl( 1-\epsilon^*(\cW;n,C + s_n)\bigr)=-\frac{1}{2V_{\max}}.
	\end{align}
\end{cor}

\begin{remark}
	Recall that our definition of c-q channels does not put any restriction on the input set. In particular, this set may be comprised of quantum states itself such that the c-q channel is just a representation of a quantum channel. Hence, as pointed out in Ref.~\cite{tomamicheltan14}, our results immediately also apply to classical communication over general image-additive channels~\cite{wolf14} as well as classical communication over quantum channels with encoders restricted to prepare separable states. We refer the reader to Corollaries 6 and 7 of Ref.~\cite{tomamicheltan14} for details. 
\end{remark}

We will split the proof of \cref{thm:moddevcode} in two, in Sect.~\ref{subsec:coding1} we will prove a lower bound on the maximum rate (`achievability'), followed in Sect.~\ref{subsec:coding2} by a corresponding the upper bound (`optimality'). For the rest of this section, we will fix the channel $\cW$, and omit any dependencies on $\cW$ from here on for notational convenience.

\subsection{Achievability}
\label{subsec:coding1}

For achievability, we will use a lower bound on the $\epsilon$-one-shot rate that is essentially due to Hayashi and Nagaoka~\cite{hayashi03} who analysed the coding problem using the information spectrum method.

\begin{lem}[Theorem 1 of Ref.~\cite{wang10}]
	\label{lem:wangrenner}
	If we have a c-q channel which maps from a finite message space $Y$ as $y\mapsto \rho^{(y)}$ , then the maximum rate with error probability at most $\epsilon$ and $1-\epsilon$, $R^*(\epsilon)$ and $R^*(1-\epsilon)$ respectively, are lower bounded
	\begin{align}
		R^*(\epsilon) &\geq \sup_{P_Y}
		D^{\epsilon/2}_{\mathrm{h}} \left( \pi_{YZ} \middle\| \pi_Y\otimes \pi_Z \right) -\log\frac{8(2-\epsilon)}{\epsilon}\\
		R^*(1-\epsilon) &\geq \sup_{P_{Y}}
		D^{1-2\epsilon}_{\mathrm{h}} \left( \pi_{YZ} \middle\| \pi_Y\otimes\pi_Z \right) -\log\frac{8(1-\epsilon)}{\epsilon}
	\end{align}
	where $\pi_{YZ}$ is the joint state of the input and output, with inputs chosen according to the distribution $P_{Y}$
	\begin{align}
		\pi_{YZ}:=\sum_{y\in Y}P_{Y}(y)\ketbra{y}{y}_{Y}\otimes \rho^{(y)}_{Z}.
	\end{align}
\end{lem}

\begin{prop}[Channel coding: Achievability]
	\label{prop:codingachievability}
	For any moderate sequence $\lbrace a_n\rbrace_n$ and error probability $\epsilon_n:=e^{-na_n^2}$, the rate is at least
	\begin{align}
		R^*\left(n,\epsilon_n\right)\geq C-\sqrt{2V_{\min}}a_n+o(a_n).
	\end{align}
	Similarly, at error probability $1-\epsilon_n$, the rate is at least
	\begin{align}
		R^*\left(n,1-\epsilon_n\right)\geq C+\sqrt{2V_{\max}}a_n+o(a_n).
	\end{align}
\end{prop}
\begin{proof}
	Let $X$ be our, possibly infinite, message space. By Lemma 3 of Ref.~\cite{tomamicheltan14}, there exists a finite subset $Y\subseteq X$, and a distribution $Q_{Y}$ thereon, such that $D(\rho\|\sigma)=C$ and $V(\rho\|\sigma)=V_{\min}$ for states
	\begin{align}
		\rho:=\sum_{y\in Y}Q_{Y}(y)\ketbra{y}{y}\otimes \rho^{(y)} 
		\qquad \text{and} \qquad \sigma:=\sum_{y\in Y}Q_{Y}(y)\ketbra{y}{y}\otimes \sum_{y'\in Y}Q_{Y}(y')\rho^{(y')}.
	\end{align}
	
	Clearly by restricting the message space we can only ever decrease the rate. By applying \cref{lem:wangrenner} to the restriction of the message space to $Y$, we can lower bound the maximum rate of the full code. Applying this reasoning to $n$ memoryless applications of our channel we find
	\begin{align}
		nR^*(n,\epsilon_n) 
		&\geq \sup_{P_{{Y}^n}} D^{\epsilon_n/2}_{\mathrm{h}} \left( \pi_{Y^nZ^n} \middle\| \pi_{Y^n}\otimes \pi_{Z^n} \right) -\log\frac{8(2-\epsilon_n)}{\epsilon_n}.
	\end{align}
	Substituting in both the error probability, which is no larger than $\epsilon_n=e^{-na_n^2}$, and a product distribution $Q_{{Y}^n}(\vec{y}):= \prod_{i=1}^n Q_{Y}(y_i)$ then we get
	\begin{align}
		R^*(n,\epsilon_n) \geq 
		\frac{1}{n}D^{\epsilon_n/2}_{\mathrm{h}} \left( \rho^{\otimes n} \middle\| \sigma^{\otimes n} \right)+\mathcal{O}(a_n^2).\label{ineq:rate}
	\end{align}
	Applying \cref{prop:HTD-inward}, we get an overall bound on the rate of
	\begin{align}
		R^*(n,\epsilon_n) \geq 
		C-\sqrt{2V_{\min}}\,a_n+o(a_n).
	\end{align}
	If instead we were to take a distribution $Q_{Y}$ such that $V(\rho\|\sigma)=V_{\max}$, then the same arguments would allow us to use \cref{prop:HTD-outward} to analogously give
	\begin{align}
		R^*(n,1-\epsilon_n) \geq 
		C+\sqrt{2V_{\max}}\,a_n+o(a_n).
	\end{align}
\end{proof}

\subsection{Optimality}
\label{subsec:coding2}
Similar to the second-order analysis of Ref.~\cite{tomamicheltan14}, we are going to do this by relating the capacity and one-shot maximum rates to geometric quantities known as the divergence radius and divergence centre. 

\begin{defn}[Divergence radius and centre]
	\label{defn:radius}
	For some set of states $\mathcal{S}_0\subseteq\mathcal{S}$, the \emph{divergence radius} $\chi(\mathcal{S}_0)$ and \emph{divergence centre} $\sigma^*(\mathcal{S}_0)$ are defined as 
	\begin{align}
		\chi(\mathcal{S}_0):=\mathop{\mathrm{inf}\vphantom{p}}_{\sigma\in\mathcal{S}} \sup_{\rho\in\mathcal{S}_0}~D(\rho\|\sigma), \qquad\qquad 
		\sigma^*(\mathcal{S}_0):=\mathop\mathrm{arg~min}\limits_{\sigma\in \mathcal{S}}\sup_{\rho\in\mathcal{S}_0}D(\rho\|\sigma).
	\end{align}
	Similarly the
	$\epsilon$-\emph{hypothesis testing divergence radius} $\chi_{\mathrm{h}}^\epsilon(\mathcal{S}_0)$ is defined as 
	\begin{align}
		\chi_{\mathrm{h}}^{\epsilon}(\mathcal{S}_0):=\mathop{\mathrm{inf}\phantom{p}}_{\sigma\in\mathcal{S}} \sup_{\rho\in\mathcal{S}_0}~D_{h}^{\epsilon}(\rho\|\sigma).
	\end{align}
\end{defn}

Whilst we have seen that the divergence radius captures the capacity of a channel, the $\epsilon$-hypothesis testing divergence radius approximates the one-shot capacity.

\begin{lem}[Proposition 5 of \cite{tomamicheltan14}]
	\label{lem:rate}
	For $\cI:=\overline{\im \cW}$, the maximum rate with error probability at most $\epsilon$, $R^*(\epsilon)$, is upper bounded as
	\begin{align}
		R^*(\epsilon)\leq \chi_{\mathrm{h}}^{2\epsilon}(I)+\log\frac{2}{1-2\epsilon}.
	\end{align}
	Similarly for an error probability $1-\epsilon$, the maximum rate is upper bounded as
	\begin{align}
		R^*(1-\epsilon)\leq \chi_{\mathrm{h}}^{1-\epsilon/2}(I)+\log\frac{2(2-\epsilon)}{\epsilon^2}.
	\end{align}
\end{lem}

If we take $\cI_n:=\overline{\im \mathcal{W}^{\otimes n}}$ to be the closure of the image of $n$ uses of this channel, then we can extend this bound on the one-shot rate to the $n$-shot rate as
\begin{align}
	nR^*(n,\epsilon_n)&\leq \chi_{\mathrm{h}}^{2\epsilon_n}(\cI_n)+\log\frac{2}{1-\epsilon_n},\\
	nR^*(n,1-\epsilon_n)&\leq \chi_{\mathrm{h}}^{2\epsilon_n}(\cI_n)+\log\frac{2(2-\epsilon_n)}{\epsilon_n^2}.
\end{align}
As we are considering memoryless c-q channels, $\cI_n$ simply consists of elementwise tensor products of $\cI$
\begin{align}
	\cI_n=\left\lbrace \bigotimes_{i=1}^n\rho_i \,\middle|\, \rho_i\in \cI \right\rbrace.
\end{align}

Once again we are going to take $a_n$ to be an arbitrary moderate sequence, and $\epsilon_n:=e^{-na_n^2}$. Expanding this out, this gives bounds on the rate of
\begin{align}
	R^*(n,\epsilon_n)&\leq \inf_{\sigma^n}\sup_{\lbrace \rho_i\rbrace\subseteq \cI}~\frac{1}{n}D^{2\epsilon_n}_{\mathrm{h}}\left(\bigotimes_{i=1}^n\rho_i \middle\|\sigma^n \right) +\frac{1}{n}\log \frac{2}{1-\epsilon_n},\label{ineq:converse}\\
	R^*(n,1-\epsilon_n)&\leq \inf_{\sigma^n}\sup_{\lbrace \rho_i\rbrace\subseteq \cI}~\frac{1}{n}D^{1-\epsilon_n/2}_{\mathrm{h}}\left(\bigotimes_{i=1}^n\rho_i \middle\|\sigma^n \right) +\frac{1}{n}\log \frac{2(2-\epsilon)}{\epsilon_n^2}.
\end{align}

A standard approach now is to pick a state $\sigma^n$, such that we can bound 
the above quantities for arbitrary sequences $\lbrace \rho_i\rbrace$ using the 
moderate deviation analysis of the hypothesis testing divergence presented in 
Sect.~\ref{sec:hypo}. To do this we need to consider two cases. The \emph{high cases} are those in which the empirical relative entropy corresponding to 
$\lbrace \rho_i\rbrace _{i=1}^n$ is close to capacity, and the \emph{low cases} 
are those in which the empirical relative entropy corresponding to $\lbrace 
\rho_i\rbrace _{i=1}^n$ is far from capacity. Specifically, for some constant 
$\gamma$ that will be chosen later, the $n$ which correspond to high and low 
cases are denoted by $H(\lbrace\rho_i\rbrace,\gamma)$ and 
$L(\lbrace\rho_i\rbrace,\gamma)$, respectively. They are defined as
\begin{align}
	H(\lbrace\rho_i\rbrace,\gamma):=\left\lbrace n~ \middle| 
	\frac{1}{n}\sum_{i=1}^{n}D(\rho_i\|\bar\rho_n)\geq C-\gamma \right\rbrace
	\quad\text{and}\quad
	L(\lbrace\rho_i\rbrace,\gamma):=\left\lbrace n~ \middle| 
	\frac{1}{n}\sum_{i=1}^{n}D(\rho_i\|\bar\rho_n)< C-\gamma \right\rbrace
\end{align}
such that $H(\lbrace\rho_i\rbrace,\gamma)$ and $L(\lbrace\rho_i\rbrace,\gamma)$ 
bipartition $\mathbb{N}$ for all $\gamma$.

Before employing a moderate deviation bound, we are going to construct a separable state $\sigma^n$ that will allow us two different moderate deviation analyses for low and high sequences, such that we can obtain the required bounds in both cases. A convenient choice of $\sigma^n$ would be  $\sigma^n={\bar\rho_n}^{\otimes n}$ where $\bar\rho_n:=\frac{1}{n}\sum_{i=1}^n\rho_i$, but the order of the infimum and supremum require $\sigma^n$ to be chosen to be independent of the sequence $\lbrace \rho_i\rbrace$. Instead we are going to construct $\sigma^n$ from a mixture of states that lie in a covering of $\mathcal{S}$, and the divergence centre $\sigma^*(\cI)$. 

The following lemma is based on a construction in Lemma II.4 of Ref.~\cite{hayden04b}.

\begin{lem}[Lemma 18 of Ref.~\cite{tomamicheltan14}]
	For every $\delta\in(0,1)$, there exists a set $\mathcal{C}^\delta\subset\mathcal{S}$ of size
	\begin{align}
	\abs{\mathcal{C}^\delta}\leq \left(\frac{20(2d+1)}{\delta}\right)^{2d^2}\left(\frac{8d(2d+1)}{\delta}+2\right)^{d-1}\leq\left(\frac{90d}{\delta^2}\right)^{2d^2}
	\end{align}
	such that, for every $\rho\in \mathcal{S}$ there exists a state $\tau\in\mathcal{C}^\delta$ such that
	\begin{align}
	D(\rho\|\tau)\leq \delta\qquad
	\textrm{and}\qquad\lambda_{\mathrm{min}}(\tau)\geq \frac{\delta}{8d(2d+1)+\delta}\geq \frac{\delta}{25d^2}.
	\end{align} 
\end{lem}

Given this covering upon states, we now want to take $\sigma^n$ to be the separable state given by a mixture over such a covering, and the divergence centre
\begin{align}
	\sigma^n(\gamma):=
	\frac{1}{2}\sigma^*(I)^{\otimes n}+
	\frac{1}{2\abs{\mathcal{C}^{\gamma/4}}}\sum_{\tau\in\mathcal{C}^{\gamma/4}} \tau^{\otimes n}.\label{eqn:sigmadef}
\end{align}
Using the inequality 
\begin{align}
	D_{\mathrm{h}}^\epsilon\bigl(\rho\,\big\|\,\mu\sigma+(1-\mu)\sigma'\bigr)\leq D^\epsilon_{\mathrm{h}}(\rho\|\sigma)-\log \mu
\end{align}
we will be able to bound divergences with respect to $\sigma^n$ by those divergences with respect to either elements of $\mathcal{C}^{\gamma/4}$, or $\sigma^*$.

We will start by considering the low case. We will see that this case only accounts for hypothesis testing relative entropies which are below the capacity by a constant amount.

\begin{lem}[Low case]
	\label{lem:low}
	For any $\gamma>0$, there exists a constant $N(\lbrace a_i\rbrace,\gamma)$ such that
	\begin{align}
		\frac{1}{n}D_{\mathrm{h}}^{\epsilon_n}\left( \bigotimes_{i=1}^n\rho_i \,\middle\|\,\sigma_n(\gamma)\right)&\leq C-\gamma/4,\\
		\frac{1}{n}D_{\mathrm{h}}^{1-\epsilon_n}\left( \bigotimes_{i=1}^n\rho_i \,\middle\|\,\sigma_n(\gamma)\right)&\leq C-\gamma/4,
	\end{align}
	for any $\lbrace\rho_i\rbrace_i\subset I$, $n\in 
	L(\lbrace\rho_i\rbrace,\gamma)$ and $n\geq N$.
\end{lem}
\begin{proof}
	We are going to start by considering the $\epsilon_n$-hypothesis testing divergence. Take $\tau_n$ to be the closest element in $\mathcal{C}^{\gamma/4}$ to $\bar{\rho}_n$, such that $D(\bar{\rho}_n\|\tau_n)\leq \gamma/4$. Splitting out the $\tau_n$ term from $\sigma_n(\gamma)$, we have
	\begin{align}
		D_{\mathrm{h}}^{\epsilon_n}\left( \bigotimes_{i=1}^n\rho_i \,\middle\|\, \sigma_n(\gamma)\right)
		&\leq D_{\mathrm{h}}^{\epsilon_n}\left( \bigotimes_{i=1}^n\rho_i \,\middle\|\, \tau_n^{\otimes n}\right)+\log 2\abs{\mathcal{C}^\gamma}\\
		&\leq D_{\mathrm{h}}^{\epsilon_n}\left( \bigotimes_{i=1}^n\rho_i \,\middle\|\, \tau_n^{\otimes n}\right)+2d^2\log\left(\frac{120d}{\gamma^2}\right).
	\end{align}
	As the final term depending on $\abs{\mathcal{C}^{\gamma/4}}$ is independent of $n$, there must exist a constant $N_1(\gamma)$ such that $2d^2\log(120d/\gamma^2)\leq n\gamma/4$ for any $n\geq N_1$, and thus that
	\begin{align}
		\frac{1}{n}D_{\mathrm{h}}^{\epsilon_n}\left( \bigotimes_{i=1}^n\rho_i \middle\| \sigma_n(\gamma)\right)
		\leq \frac{1}{n}D_{\mathrm{h}}^{\epsilon_n}\left( \bigotimes_{i=1}^n\rho_i \middle\| \tau_n^{\otimes n}\right)+\gamma/4.
	\end{align}
	
	Applying \cref{prop:HTD-outward} to the $\epsilon_n$-hypothesis testing relative entropy with respect to $\tau_n$ we get that there exists an $N_2(\lbrace a_i\rbrace,\gamma)$ such that 
	\begin{align}
		\frac{1}{n}D_{\mathrm{h}}^{\epsilon_n}\left( \bigotimes_{i=1}^n\rho_i \middle\| \tau_n^{\otimes} \right)\leq \frac{1}{n}\sum_{i=1}^nD(\rho_i\|\tau_n)+\gamma/4,
	\end{align}
	for any $n\geq N_2$. As for the divergence terms given with respect to $\tau_n$, we can rearrange them in terms of divergences relative to the sequence mean $\bar\rho_n$ using the information geometric Pythagorean theorem, yielding
	\begin{align}
		\sum_{i=1}^{n}D(\rho_i\|\tau_n)
		&=\sum_{i=1}^{n}\Tr \rho_i(\log\rho_i-\log\bar\rho_n)+\sum_{i=1}^{n}\Tr\rho_i(\log\bar\rho_n-\log \tau_n)\label{eqn:line2}\\
		&=\sum_{i=1}^{n}D(\rho_i\|\bar\rho_n)+nD(\bar\rho_n\|\tau_n)\\
		&\leq\sum_{i=1}^{n}D(\rho_i\|\bar\rho_n)+n\gamma/4.
	\end{align}
	
	If we let $N(\lbrace a_i\rbrace, \gamma):=\max\lbrace N_1,N_2\rbrace$, then pulling the above results together we see that for any $n\geq N$ 
	\begin{align}
		\frac{1}{n}D_{\mathrm{h}}^{\epsilon_n}\left( \bigotimes_{i=1}^n\rho_i \middle\| \sigma_n(\gamma) \right)
		&\leq\frac{1}{n}D_{\mathrm{h}}^{\epsilon_n}\left( \bigotimes_{i=1}^n\rho_i \middle\| \tau_n^{\otimes n}\right)+\gamma/4\\
		&\leq\frac{1}{n}\sum_{i=1}^nD(\rho_i\|\tau_n)+2\gamma/4\\
		&\leq\frac{1}{n}\sum_{i=1}^nD(\rho_i\|\bar{\rho}_n)+3\gamma/4.
	\end{align}
	Finally, since $n\in L(\lbrace\rho_i\rbrace,\gamma)$ the average relative 
	entropy is bounded away from capacity, and we arrive at the bound:
	\begin{align}
		\frac{1}{n}D_{\mathrm{h}}^{\epsilon_n}\left( \bigotimes_{i=1}^n\rho_i \middle\| \sigma_n(\gamma) \right)\leq C-\gamma/4.
	\end{align}
	
	As we only relied on \cref{prop:HTD-outward} to bound the regularised 
	hypothesis testing divergence to within a constant of the average relative 
	entropy, we could perform a similar analysis for the 
	$(1-\epsilon_n)$-hypothesis testing divergence using \cref{prop:HTD-inward} 
	instead, which gives
	\begin{align}
		\frac{1}{n}D_{\mathrm{h}}^{1-\epsilon_n}\left( \bigotimes_{i=1}^n\rho_i \middle\| \sigma_n(\gamma) \right)\leq C-\gamma/4.
	\end{align}
\end{proof}

Now that we have dealt with cases far from capacity, we turn our attention to the high cases.

\begin{lem}[High case]
	\label{lem:high}
	For any $\eta>0$, there exist constants $\Gamma(\eta)$ and $N(\lbrace a_i\rbrace,\eta)$, such that
	\begin{align}
		\frac{1}{n}D_{\mathrm{h}}^{\epsilon_n}\left( \bigotimes_{i=1}^n\rho_i \,\middle\|\,\sigma_n(\Gamma)\right)\leq C-\sqrt{2V_{\min}}\,a_n+\eta a_n
	\end{align}
	for any $\lbrace\rho_i\rbrace_i\subset I$, $n\in 
	H(\lbrace\rho_i\rbrace,\Gamma)$ and $n\geq N$. Similarly, the 
	$(1-\epsilon_n)$-hypothesis testing relative entropy is bounded
	\begin{align}
		\frac{1}{n}D_{\mathrm{h}}^{1-\epsilon_n}\left( \bigotimes_{i=1}^n\rho_i \,\middle\|\,\sigma_n(\Gamma)\right)\leq C+\sqrt{2V_{\max}}a_n+\eta a_n.
	\end{align}
\end{lem}
\begin{proof}
	Splitting out the $\sigma^*$ factor within $\sigma_n(\gamma)$ gives
	\begin{align}
		D_{\mathrm{h}}^{\epsilon_n}\left( \bigotimes_{i=1}^n\rho_i \,\middle\|\,\sigma_n(\gamma)\right)&\leq D_{\mathrm{h}}^{\epsilon_n}\left( \bigotimes_{i=1}^n\rho_i \,\middle\|\,{\sigma^*}^{\otimes n}\right)+\log 2,\\
		D_{\mathrm{h}}^{1-\epsilon_n}\left( \bigotimes_{i=1}^n\rho_i \,\middle\|\,\sigma_n(\gamma)\right)&\leq D_{\mathrm{h}}^{1-\epsilon_n}\left( \bigotimes_{i=1}^n\rho_i \,\middle\|\,{\sigma^*}^{\otimes n}\right)+\log 2.
	\end{align}
	As $\frac{1}{n}\log 2=o(a_n)$, there exists an $N_1(\lbrace a_i\rbrace)$ such that $n\geq N_1$ implies 
	\begin{align}
		D_{\mathrm{h}}^{\epsilon_n}\left( \bigotimes_{i=1}^n\rho_i \,\middle\|\,\sigma_n(\gamma)\right)&\leq D_{\mathrm{h}}^{\epsilon_n}\left( \bigotimes_{i=1}^n\rho_i \,\middle\|\,{\sigma^*}^{\otimes n}\right)+\eta a_n/3,\\
		D_{\mathrm{h}}^{1-\epsilon_n}\left( \bigotimes_{i=1}^n\rho_i \,\middle\|\,\sigma_n(\gamma)\right)&\leq D_{\mathrm{h}}^{1-\epsilon_n}\left( \bigotimes_{i=1}^n\rho_i \,\middle\|\,{\sigma^*}^{\otimes n}\right)+\eta a_n/3.
	\end{align}
	We now wish to employ a moderate deviation result. We will start by addressing the $\epsilon_n$-hypothesis testing divergence. For the weaker bound of \cref{prop:HTD-outward} we will have no required bounds on $\frac{1}{n}\sum_{i=1}^{n}V(\rho_i\|\sigma^*)$, but for the stronger bound we will need a uniform lower bound.
	
	If $V_{\min}\leq \eta^2/18$, then the weakened bound of \cref{prop:HTD-outward} is sufficient, giving an $N_2(\lbrace a_n\rbrace,\eta)$ such that $n\geq N_2$ implies
	\begin{align}
		D_{\mathrm{h}}^{\epsilon_n}\left( \bigotimes_{i=1}^n\rho_i \,\middle\|\,\sigma_n(\gamma)\right)&\leq\frac{1}{n}D_{\mathrm{h}}^{\epsilon_n}\left( \bigotimes_{i=1}^n\rho_i \,\middle\|\,{\sigma^*}^{\otimes n}\right)+\eta a_n/3\\
		&\leq \frac{1}{n}\sum_{i=1}^{n}D(\rho_i\|\sigma^*)+2\eta a_n/3\\
		&\leq \frac{1}{n}\sum_{i=1}^{n}D(\rho_i\|\sigma^*)-\sqrt{2V_{\min}}\,a_n+\eta a_n\\
		&\leq C-\sqrt{2V_{\min}}\,a_n+\eta a_n.
	\end{align}
	
	Next we need to consider the case where $V_{\min}>\eta^2/18$. To do this, we will need to establish a lower bound on $\frac{1}{n}\sum_{i=1}^{n}V(\rho_i\|\sigma^*)$, which places it near $V_{\min}$. The min-dispersion is defined for distributions which exactly achieve capacity; we will now consider an analogous quantity for distributions which are \emph{near} capacity. Specifically
	\begin{align}
		V_{\min}(\gamma):=\inf_{P\in\cP(\cI)}\left\lbrace 
		\int\mathrm{d}P(\rho)~V\left(\rho\middle\|\sigma^*\right) 
		~\middle|~
		\int\mathrm{d}P(\rho)~D\left(\rho~\middle\|\int\mathrm{d}P(\rho')~\rho'\right)\geq C-\gamma
		\right\rbrace.
	\end{align}
	By definition of the channel dispersion we have that $V_{\min}(0)=V_{\min}$. By Lemma 22 of Ref.~\cite{tomamicheltan14} we can strengthen this to $\lim_{\gamma\to 0^+}V_{\min}(\gamma)=V_{\min}$, and so for any $\eta>0$ there must exist a constant $\Gamma(\eta)$ such that 
	\begin{align}
		\sqrt{2V_{\min}(\Gamma)}\geq \sqrt{2V_{\min}}-\eta/3.\label{eqn:V}
	\end{align} 
	As $V_{\min}\geq \eta^2/18$, this implies that $V_{\min}(\Gamma)>0$.
	
	Next, let $P_n$ be the empirical distribution corresponding to the set 
	$\lbrace\rho_i\rbrace_{i=1}^n$, i.e.\ 
	$P_n(\rho):=\frac{1}{n}\sum_{i=1}^{n}\delta(\rho-\rho_i)$. For all $n\in 
	H(\lbrace\rho_i\rbrace,\Gamma)$, these distributions are near capacity
	\begin{align}
		\int\mathrm{d}P_n(\rho)~ D\left(\rho~\middle\|\int\mathrm{d}P_n(\rho')~\rho'\right)=\frac{1}{n}\sum_{i=1}^{n}D(\rho_i\|\bar{\rho}_n)\geq C-\Gamma,
	\end{align}
	and so we can lower bound the average variance with respect to the divergence centre
	\begin{align}
		\frac{1}{n}\sum_{i=1}^{n}V(\rho_i\|\sigma^*)=\int\mathrm{d}P(\rho)~ V(\rho\|\sigma^*)\geq V_{\min}(\Gamma)>0.
	\end{align}
	Using this lower bound, we can apply the stronger bound from 
	\cref{prop:HTD-outward} to give a constant $N_3(\lbrace a_i\rbrace,\eta)$, 
	such that, for every $n\in H(\lbrace\rho_i\rbrace,\Gamma)$ and $n\geq N_3$, 
	the hypothesis testing divergence is upper bounded
	\begin{align}
		D_{\mathrm{h}}^{\epsilon_n}\left( \bigotimes_{i=1}^n\rho_i \,\middle\|\,\sigma_n(\gamma)\right)&\leq\frac{1}{n}D_{\mathrm{h}}^{\epsilon_n}\left( \bigotimes_{i=1}^n\rho_i \,\middle\|\,{\sigma^*}^{\otimes n}\right)+\eta a_n/3\\
		&\leq \frac{1}{n}\sum_{i=1}^{n}D(\rho_i\|\sigma^*)-\sqrt{\frac{2}{n}\sum_{i=1}^{n}V(\rho_i\|\sigma^*)}a_n+2\eta a_n/3\\
		&\leq \frac{1}{n}\sum_{i=1}^{n}D(\rho_i\|\sigma^*)-\sqrt{2V_{\min}(\Gamma)}\,a_n+2\eta a_n/3\\
		&\leq C-\sqrt{2V_{\min}}\,a_n+\eta a_n.
	\end{align}
	
	Performing a similar argument for $V_{\max}$, we construct a function
	\begin{align}
		V_{\max}(\gamma):=\sup_{P\in\cP(\cI)}\left\lbrace 
		\int\mathrm{d}P(\rho)~V\left(\rho\middle\|\sigma^*\right) 
		~\middle|~
		\int\mathrm{d}P(\rho)~D\left(\rho~\middle\|\int\mathrm{d}P(\rho')~\rho'\right)\geq C-\gamma
		\right\rbrace,
	\end{align}
	and define a $\Gamma$ such that
	\begin{align}
	\sqrt{2V_{\max}(\Gamma)}\leq \sqrt{2V_{\max}}+\eta/3.
	\end{align} 
	Following through the rest of the argument, and employing \cref{prop:HTD-inward}, we also get a bound on the $(1-\epsilon_n)$-hypothesis testing divergence
	\begin{align}
		\frac{1}{n}D_{\mathrm{h}}^{1-\epsilon_n}\left( \bigotimes_{i=1}^n\rho_i \,\middle\|\,\sigma_n(\gamma)\right)
		&\leq C+\sqrt{2V_{\max}}a_n+\eta a_n.
	\end{align} 
\end{proof}

\begin{prop}[Channel coding: Optimality]
	\label{prop:codingconverse}
	For any moderate sequence $\lbrace a_n\rbrace_n$ and error probability $\epsilon_n:=e^{-na_n^2}$, the rate is upper bounded as
	\begin{align}
		R^*\left(n,\epsilon_n\right)\leq C-\sqrt{2V_{\min}}\,a_n+o(a_n).
	\end{align}
	For error probability $(1-\epsilon_n)$ the rate is similarly upper bound as
	\begin{align}
		R^*\left(n,1-\epsilon_n\right)\leq C+\sqrt{2V_{\max}}\,a_n+o(a_n).
	\end{align}
\end{prop}
\begin{proof}
	Applying Lemmas \ref{lem:low} and \ref{lem:high}, we get that there exist constants $\Gamma(\eta)$ and $N_1(\lbrace a_i\rbrace,\eta)$ such that
	\begin{align}
		\frac{1}{n}D_{\mathrm{h}}^{\epsilon_n}\left( \bigotimes_{i=1}^n\rho_i \,\middle\|\, \sigma_n(\Gamma)\right)\leq 
		\begin{dcases}
			C-\Gamma/4 & n\in L(\lbrace\rho_i\rbrace,\Gamma)\\
			C-\sqrt{2V_{\min}}\,a_n+\eta a_n & n\in H(\lbrace\rho_i\rbrace,\Gamma)
		\end{dcases}
	\end{align}
	for any $n\geq N_1$. As $\Gamma$ is a constant, there must exist some $N_2(\lbrace a_i\rbrace,\eta)$ such that $\Gamma/4\geq \sqrt{2V_{\min}}a_n$. As such, for any $n\geq \max\lbrace N_1,N_2\rbrace$, high or low, we have 
	\begin{align}
		\frac{1}{n}D_{\mathrm{h}}^{\epsilon_n}\left( \bigotimes_{i=1}^n\rho_i \,\middle\|\, \sigma_n(\Gamma)\right)\leq C-\sqrt{2V_{\min}}\,a_n+\eta a_n.
	\end{align}
	Pulling this bound back to Eq.~\ref{ineq:converse}, we have
	\begin{align}
		R^*(n,\epsilon_n)
		&\leq \sup_{\lbrace \rho_i\rbrace\subseteq I}\frac{1}{n}D^{2\epsilon_n}_{\mathrm{h}}\left(\bigotimes_{i=1}^n\rho_i 
		\middle\|\sigma^n(\Gamma) \right) +\frac{1}{n}\log 
		\frac{2}{1-\epsilon_n}\\
		&\leq C-\sqrt{2V_{\min}}a_n+\eta a_n+\frac{1}{n}\log \frac{2}{1-\epsilon_n}.
	\end{align}
	Finally, noting that $1/n=o(a_n)$, there must exist a constant $N_3 (\lbrace a_i\rbrace,\eta)$ such that $n\geq N_3$ implies
	\begin{align}
	\frac{1}{n}\log\frac{2}{1-\epsilon_n}\leq \eta a_n.
	\end{align}
	We can therefore conclude that, for $n\geq \max\lbrace N_1,N_2,N_3\rbrace $, we get the overall upper bound
	\begin{align}
	R^*(n,\epsilon_n)\leq C-\sqrt{2V_{\min}}\,a_n+2\eta a_n.
	\end{align}
	As this is true for arbitrary $\eta>0$, we can take $\eta \searrow0$ and conclude 
	\begin{align}
	R^*(n,\epsilon_n)\leq C-\sqrt{2V_{\min}}a_n+o(a_n)
	\end{align}
	as required. A similar analysis for the $(1-\epsilon_n)$-error regime shows
	\begin{align}
	R^*\left(n,1-\epsilon_n\right)\leq C+\sqrt{2V_{\max}}a_n+o(a_n).
	\end{align}
\end{proof}


\section{Conclusion}

The main result of this paper is to give a second order approximation of the non-asymptotic fundamental limit for classical information transmission over a quantum channel in the moderate deviations regime, as in Eqs.~\ref{eq:moderatesecondorder1} and \ref{eq:moderatesecondorder2}:
\begin{align}
	\frac{1}{n} \log M^*(\cW;n, \eps_n) &= C(\cW) - \sqrt{2 V_{\min}(\cW)}\, x_n + o(x_n) \,,\\
	\frac{1}{n} \log M^*(\cW;n,1- \eps_n) &= C(\cW) + \sqrt{2 V_{\max}(\cW)}\, x_n + o(x_n) \,.
\end{align}
Along the lines of third and fourth order approximations for classical channel coding in the fixed error regime (see, e.g., Refs.~\cite{polyanskiythesis10,tomamicheltan12,moulin12}), a natural question to ask is whether we can expand this further and resolve the term $o(x_n)$. A preliminary investigation suggests the conjecture that $o(x_n) = O(x_n^2) + O(\log n)$ and that at least some of the implicit constants can be determined precisely. We leave this for future work.

Due to the central importance of binary asymmetric quantum hypothesis testing 
we expect our techniques to have applications also to other quantum channel 
coding tasks. In particular, source coding~\cite{datta15,leditzky16}, 
entanglement-assisted classical coding~\cite{datta14} as well as 
quantum~\cite{tomamichel16} and private coding~\cite{wilde16} over quantum 
channels have recently been analysed in the small deviations regime by relating 
the problem to quantum hypothesis testing. An extension of these results to 
moderate deviations using our techniques thus appears feasible.

\paragraph*{Acknowledgements.} MT is funded by an Australian Research Council Discovery Early Career Researcher Award (DECRA) fellowship (Grant Nos. CE110001013, DE160100821). Both MT and CTC acknowledge support from the ARC Centre of Excellence for Engineered Quantum Systems (EQuS). We also thank Hao-Chung Cheng and Min-Hsiu Hsieh for useful discussions and insightful comments.



\begin{thebibliography}{10}
	\bibitem{holevo98}
	A.~S. Holevo, ``{The Capacity of the Quantum Channel with General Signal
		States},'' {\em \href{http://dx.doi.org/10.1109/18.651037}{IEEE Transactions
			on Information Theory}} {\bfseries 44}, 269--273, \href{http://arxiv.org/abs/quant-ph/9611023}{arXiv:quant-ph/9611023}, (1998).
	
	\bibitem{holevo73b}
	A.~S. Holevo, ``{Bounds for the Quantity of Information Transmitted by a
		Quantum Communication Channel},'' {\em Problems of Information Transmission}
	{\bfseries 9}, 3, 177--183,  (1973).
	
	\bibitem{schumacher97}
	B.~Schumacher and M.~Westmoreland, ``{Sending Classical Information via Noisy
		Quantum Channels},'' {\em
		\href{http://dx.doi.org/10.1103/PhysRevA.56.131}{Physical Review A}}
	{\bfseries 56}, 131--138,  (1997).
	
	\bibitem{holevo00}
	A.~Holevo, ``{Reliability function of general classical-quantum channel},''
	{\em \href{http://dx.doi.org/10.1109/18.868501}{IEEE Transactions on
			Information Theory}} {\bfseries 46}, 2256--2261,
	\href{http://arxiv.org/abs/quant-ph/9907087}{arXiv:quant-ph/9907087},
	(2000).
	
	\bibitem{hayashi07}
	M.~Hayashi, ``{Error Exponent in Asymmetric Quantum Hypothesis Testing and Its
		Application to Classical-Quantum Channel coding},'' {\em
		\href{http://dx.doi.org/10.1103/PhysRevA.76.062301}{Physical Review A}}
	{\bfseries 76}, 062301,
	\href{http://arxiv.org/abs/quant-ph/0611013}{arXiv:quant-ph/0611013},
	(2006).
	
	\bibitem{dalai13}
	M.~Dalai, ``{Lower Bounds on the Probability of Error for Classical and
		Classical-Quantum Channels},'' {\em
		\href{http://dx.doi.org/10.1109/TIT.2013.2283794}{IEEE Transactions on
			Information Theory}} {\bfseries 59}, 8027--8056,
	\href{http://arxiv.org/abs/1201.5411}{arXiv:1201.5411},  (2013).
	
	\bibitem{tomamicheltan14}
	M.~Tomamichel and V.~Y.~F. Tan, ``{Second-Order Asymptotics for the Classical
		Capacity of Image-Additive Quantum Channels},'' {\em
		\href{http://dx.doi.org/10.1007/s00220-015-2382-0}{Communications in
			Mathematical Physics}} {\bfseries 338}, 103--137,
	\href{http://arxiv.org/abs/1308.6503}{arXiv:1308.6503},  (2015).
	
	\bibitem{dembo98}
	A.~Dembo and O.~Zeitouni, {\em \href{http://dx.doi.org/10.1007/978-3-642-03311-7
		}{Large Deviations Techniques and
				Applications}},
	\newblock Stochastic Modelling and Applied Probability, Springer,
	(1998).
	
	\bibitem{umegaki62}
	H.~Umegaki, ``{Conditional expectation in an operator algebra. IV. Entropy and
		information},'' {\em \href{http://dx.doi.org/10.2996/kmj/1138844604}{Kodai
			Mathematical Seminar Reports}} {\bfseries 14}, 2, 59--85,  (1962).
	
	\bibitem{tomamichel12}
	M.~Tomamichel and M.~Hayashi, ``{A Hierarchy of Information Quantities for
		Finite Block Length Analysis of Quantum Tasks},'' {\em
		\href{http://dx.doi.org/10.1109/TIT.2013.2276628}{IEEE Transactions on
			Information Theory}} {\bfseries 59}, 7693--7710,
	\href{http://arxiv.org/abs/1208.1478}{arXiv:1208.1478},  (2013).
	
	\bibitem{li12}
	K.~Li, ``{Second-Order Asymptotics for Quantum Hypothesis Testing},'' {\em
		\href{http://dx.doi.org/10.1214/13-AOS1185}{Annals of Statistics}} {\bfseries
		42}, 171--189, \href{http://arxiv.org/abs/1208.1400}{arXiv:1208.1400},
	(2014).
	
	\bibitem{altug14}
	Y.~Altug and A.~B. Wagner, ``{Moderate Deviations in Channel Coding},'' {\em
		\href{http://dx.doi.org/10.1109/TIT.2014.2323418}{IEEE Transactions on
			Information Theory}} {\bfseries 60}, 4417--4426,
	\href{http://arxiv.org/abs/1208.1924}{arXiv:1208.1924},  (2014).
	
	\bibitem{polyanskiy10c}
	Y.~Polyanskiy and S.~Verdu, ``{Channel dispersion and moderate deviations
		limits for memoryless channels},''
	\href{http://dx.doi.org/10.1109/ALLERTON.2010.5707068}{ {\em 2010 48th
			Annual Allerton Conference on Communication, Control, and Computing
			(Allerton)}}, 1334--1339, IEEE,  (2010).
	
	\bibitem{wang10}
	L.~Wang and R.~Renner, ``{One-Shot Classical-Quantum Capacity and Hypothesis
		Testing},'' {\em
		\href{http://dx.doi.org/10.1103/PhysRevLett.108.200501}{Physical Review
			Letters}} {\bfseries 108}, 200501,
	\href{http://arxiv.org/abs/1007.5456}{arXiv:1007.5456},  (2012).
	
	\bibitem{cheng17}
	H.-C. Cheng and M.-H. Hsieh, ``{Moderate Deviation Analysis for
		Classical-Quantum Channels and Quantum Hypothesis Testing},''
	\href{http://arxiv.org/abs/1701.03195}{arXiv:1701.03195},  (2016).
	
	\bibitem{hoeffding65}
	W.~Hoeffding, ``{Asymptotically Optimal Tests for Multinomial Distributions},''
	\href{http://dx.doi.org/10.1007/978-1-4612-0865-5_28}{\em Annals of Mathematical Statistics} {\bfseries 36}, 2, 369--401,  (1965).
	
	\bibitem{gallager68}
	R.~G. Gallager, \href{http://dx.doi.org/10.1007/978-3-7091-2945-6}{\em {\em {Information Theory and Reliable Communication}}},
	\newblock Wiley, (1968).
	
	\bibitem{csiszar11}
	I.~Csisz{\'{a}}r and J.~K{\"{o}}rner, \href{http://dx.doi.org/10.1017/CBO9780511921889}{\em {\em {Information Theory: Coding
				Theorems for Discrete Memoryless Systems}}},
	\newblock Cambridge University Press, (2011).
	
	\bibitem{nagaoka06}
	H.~Nagaoka, ``{The Converse Part of The Theorem for Quantum Hoeffding Bound},''
	\href{http://arxiv.org/abs/quant-ph/0611289}{arXiv:quant-ph/0611289},
	(2006).
	
	\bibitem{Sas11}
	I.~Sason, ``{Moderate deviations analysis of binary hypothesis testing},''
	\href{http://dx.doi.org/10.1109/ISIT.2012.6284675}{ {\em 2012 IEEE
			International Symposium on Information Theory Proceedings}}, 821--825, IEEE, \href{http://arxiv.org/abs/1111.1995}{arXiv:1111.1995}, (2012).
	
	\bibitem{strassen62}
	V.~Strassen, ``{Asymptotische Absch{}{\"{a}}{}tzungen in Shannons
		Informationstheorie},'' {\em Trans. Third Prague Conference on Information
		Theory}, Prague, 689--723,  (1962).
	
	\bibitem{hayashi09}
	M.~Hayashi, ``{Information Spectrum Approach to Second-Order Coding Rate in
		Channel Coding},'' {\em
		\href{http://dx.doi.org/10.1109/TIT.2009.2030478}{IEEE Transactions on
			Information Theory}} {\bfseries 55}, 4947--4966,
	\href{http://arxiv.org/abs/0801.2242}{arXiv:0801.2242},  (2009).
	
	\bibitem{polyanskiy10}
	Y.~Polyanskiy, H.~V. Poor, and S.~Verd{\'{u}}, ``{Channel Coding Rate in the
		Finite Blocklength Regime},'' {\em
		\href{http://dx.doi.org/10.1109/TIT.2010.2043769}{IEEE Transactions on
			Information Theory}} {\bfseries 56}, 2307--2359,  (2010).
	
	\bibitem{csiszar71}
	I.~Csisz{\'{a}}r and G.~Longo, ``{On the Error Exponent for Source Coding and
		for Testing Simple Statistical Hypotheses},'' {\em Studia Scientiarum
		Mathematicarum Hungarica} {\bfseries 6}, 181--191,  (1971).
	
	\bibitem{han89}
	T.~S. Han and K.~Kobayashi, ``{The Strong Converse Theorem for Hypothesis
		Testing},'' {\em \href{http://dx.doi.org/10.1109/18.42188}{IEEE Transactions
			on Information Theory}} {\bfseries 35}, 178--180,  (1989).
	
	\bibitem{arimoto73}
	S.~Arimoto, ``{On the converse to the coding theorem for discrete memoryless
		channels},'' {\em \href{http://dx.doi.org/10.1109/TIT.1973.1055007}{IEEE
			Transactions on Information Theory}} {\bfseries 19}, 357--359,  (1973).
	
	\bibitem{dueck79}
	G.~Dueck and J.~Korner, ``{Reliability function of a discrete memoryless
		channel at rates above capacity (Corresp.)},'' {\em
		\href{http://dx.doi.org/10.1109/TIT.1979.1056003}{IEEE Transactions on
			Information Theory}} {\bfseries 25}, 82--85,  (1979).
	
	\bibitem{mosonyiogawa13}
	M.~Mosonyi and T.~Ogawa, ``{Quantum Hypothesis Testing and the Operational
		Interpretation of the Quantum R{\'{e}}nyi Relative Entropies},'' {\em
		\href{http://dx.doi.org/10.1007/s00220-014-2248-x}{Communications in
			Mathematical Physics}} {\bfseries 334}, 1617--1648,
	\href{http://arxiv.org/abs/1309.3228}{arXiv:1309.3228},  (2015).
	
	\bibitem{mosonyi14}
	M.~Mosonyi and T.~Ogawa, ``{Two approaches to obtain the strong converse
		exponent of quantum hypothesis testing for general sequences of quantum
		states},'' {\em \href{http://dx.doi.org/10.1109/TIT.2015.2489259}{IEEE
			Transactions on Information Theory}} {\bfseries 61}, 6975--6994,
	\href{http://arxiv.org/abs/1407.3567}{arXiv:1407.3567},  (2014).
	
	\bibitem{mosonyi14-2}
	M.~Mosonyi and T.~Ogawa, ``{Strong Converse Exponent for Classical-Quantum
		Channel Coding},'' \href{http://arxiv.org/abs/1409.3562}{arXiv:1409.3562},
	(2014).
	
	\bibitem{partha67}
	K.~R. Parthasarathy, {\em {\em {Probability Measures on Metric Spaces}}}.
	\newblock Academic Press, New York and London,  (1967).
	
	\bibitem{winter99}
	A.~Winter, ``{Coding Theorem and Strong Converse for Quantum Channels},'' {\em
		\href{http://dx.doi.org/10.1109/18.796385}{IEEE Transactions on Information
			Theory}} {\bfseries 45}, 2481--2485,
	\href{http://arxiv.org/abs/1409.2536}{arXiv:1409.2536},  (2014).
	
	\bibitem{ogawa99}
	T.~Ogawa and H.~Nagaoka, ``{Strong converse to the quantum channel coding
		theorem},'' {\em \href{http://dx.doi.org/10.1109/18.796386}{IEEE Transactions
			on Information Theory}} {\bfseries 45}, 2486--2489,
	\href{http://arxiv.org/abs/quant-ph/9808063}{arXiv:quant-ph/9808063},
	(1999).
	
	\bibitem{schumacher01}
	B.~Schumacher and M.~D. Westmoreland, ``{Optimal signal ensembles},'' {\em
		\href{http://dx.doi.org/10.1103/PhysRevA.63.022308}{Physical Review A}}
	{\bfseries 63}, 022308,
	\href{http://arxiv.org/abs/quant-ph/9912122}{arXiv:quant-ph/9912122},
	(1999).
	
	\bibitem{dupuis12}
	F.~Dupuis, L.~Kraemer, P.~Faist, J.~M. Renes, and R.~Renner, ``{Generalized
		Entropies},'' \href{http://dx.doi.org/10.1142/9789814449243_0008}{{\em
			Proc. XVIIth International Congress on Mathematical Physics}}, 134--153,
	\href{http://arxiv.org/abs/1211.3141}{arXiv:1211.3141}, (2012).
	
	\bibitem{petz86}
	D.~Petz, ``{Quasi-entropies for finite quantum systems},'' {\em
		\href{http://dx.doi.org/10.1016/0034-4877(86)90067-4}{Reports on Mathematical
			Physics}} {\bfseries 23}, 57--65,  (1986).
	
	\bibitem{lintomamichel14}
	M.~S. Lin and M.~Tomamichel, ``{Investigating Properties of a Family of Quantum
		Renyi Divergences},'' {\em
		\href{http://dx.doi.org/10.1007/s11128-015-0935-y}{Quantum Information
			Processing}} {\bfseries 14}, 1501--1512,
	\href{http://arxiv.org/abs/1408.6897}{arXiv:1408.6897},  (2014).
	
	\bibitem{NussbaumSzkola2009}
	M.~Nussbaum and A.~Szko{\l}a, ``{The Chernoff lower bound for symmetric quantum
		hypothesis testing},'' {\em \href{http://dx.doi.org/10.1214/08-AOS593}{The
			Annals of Statistics}} {\bfseries 37}, 1040--1057,
	\href{http://arxiv.org/abs/quant-ph/0607216}{arXiv:quant-ph/0607216},
	(2009).
	
	\bibitem{wolf14}
	M.~Fukuda, I.~Nechita, and M.~M. Wolf, ``{Quantum Channels With Polytopic
		Images and Image Additivity},'' {\em
		\href{http://dx.doi.org/10.1109/TIT.2015.2401397}{IEEE Transactions on
			Information Theory}} {\bfseries 61}, 1851--1859,
	\href{http://arxiv.org/abs/1408.2340}{arXiv:1408.2340},  (2015).
	
	\bibitem{hayashi03}
	M.~Hayashi and H.~Nagaoka, ``{General Formulas for Capacity of
		Classical-Quantum Channels},'' {\em
		\href{http://dx.doi.org/10.1109/TIT.2003.813556}{IEEE Transactions on
			Information Theory}} {\bfseries 49}, 1753--1768,
	\href{http://arxiv.org/abs/quant-ph/0206186}{arXiv:quant-ph/0206186},
	(2003).
	
	\bibitem{hayden04b}
	P.~Hayden, D.~Leung, P.~W. Shor, and A.~Winter, ``{Randomizing Quantum States:
		Constructions and Applications},'' {\em
		\href{http://dx.doi.org/10.1007/s00220-004-1087-6}{Communications in
			Mathematical Physics}} {\bfseries 250}, 371--391,
	\href{http://arxiv.org/abs/quant-ph/03071}{arXiv:quant-ph/03071},  (2004).
	
	\bibitem{polyanskiythesis10}
	Y.~Polyanskiy, {\em {Channel Coding: Non-Asymptotic Fundamental Limits}},
	\newblock PhD thesis, Princeton University,  (2010).
	
	\bibitem{tomamicheltan12}
	M.~Tomamichel and V.~Y.~F. Tan, ``{A Tight Upper Bound for the Third-Order
		Asymptotics for Most Discrete Memoryless Channels},'' {\em
		\href{http://dx.doi.org/10.1109/TIT.2013.2276077}{IEEE Transactions on
			Information Theory}} {\bfseries 59}, 7041--7051,
	\href{http://arxiv.org/abs/1212.3689}{arXiv:1212.3689},  (2013).
	
	\bibitem{moulin12}
	P.~Moulin, ``{The Log-Volume of Optimal Codes for Memoryless Channels,
		Asymptotically Within A Few Nats},'' {\em
		\href{http://dx.doi.org/10.1109/TIT.2016.2643681}{IEEE Transactions on
			Information Theory}}, {\bfseries 63}, 2278 -- 2313,
	\href{http://arxiv.org/abs/1311.0181}{arXiv:1311.0181},  (2017).
	
	\bibitem{datta15}
	N.~Datta and F.~Leditzky, ``{Second-Order Asymptotics for Source Coding, Dense
		Coding, and Pure-State Entanglement Conversions},'' {\em
		\href{http://dx.doi.org/10.1109/TIT.2014.2366994}{IEEE Transactions on
			Information Theory}} {\bfseries 61}, 582--608,
	\href{http://arxiv.org/abs/1403.2543}{arXiv:1403.2543},  (2015).
	
	\bibitem{leditzky16}
	F.~Leditzky and N.~Datta, ``{Second-Order Asymptotics of Visible Mixed Quantum
		Source Coding via Universal Codes},'' {\em
		\href{http://dx.doi.org/10.1109/TIT.2016.2571662}{IEEE Transactions on
			Information Theory}} {\bfseries 62}, 4347--4355,
	\href{http://arxiv.org/abs/1407.6616}{arXiv:1407.6616},  (2016).
	
	\bibitem{datta14}
	N.~Datta, M.~Tomamichel, and M.~M. Wilde, ``{On the second-order asymptotics
		for entanglement-assisted communication},'' {\em
		\href{http://dx.doi.org/10.1007/s11128-016-1272-5}{Quantum Information
			Processing}} {\bfseries 15}, 2569--2591,
	\href{http://arxiv.org/abs/1405.1797}{arXiv:1405.1797},  (2016).
	
	\bibitem{tomamichel16}
	M.~Tomamichel, M.~Berta, and J.~M. Renes, ``{Quantum coding with finite
		resources},'' {\em \href{http://dx.doi.org/10.1038/ncomms11419}{Nature
			Communications}} {\bfseries 7}:11419,
	\href{http://arxiv.org/abs/1504.04617}{arXiv:1504.04617},  (2016).
	
	\bibitem{wilde16}
	M.~M. Wilde, M.~Tomamichel, and M.~Berta, ``{Converse bounds for private
		communication over quantum channels},'' {\em
		\href{http://dx.doi.org/10.1109/TIT.2017.2648825}{IEEE Transactions on
			Information Theory}}, {\bfseries 63}, 1792--1817,
	\href{http://arxiv.org/abs/1602.08898}{arXiv:1602.08898},  (2017).
	
	\bibitem{Rozovsky2002}
	L.~V. Rozovsky, ``{Estimate from Below for Large-Deviation Probabilities of a
		Sum of Independent Random Variables with Finite Variances},'' {\em
		\href{http://dx.doi.org/10.1023/A:1014589618720}{Journal of Mathematical
			Sciences}} {\bfseries 109}, 6, 2192--2209,  (2002).
	
	\bibitem{lee16}
	S.-H. Lee, V.~Y.~F. Tan, and A.~Khisti, ``{Streaming Data Transmission in the
		Moderate Deviations and Central Limit Regimes},'' {\em
		\href{http://dx.doi.org/10.1109/TIT.2016.2619713}{IEEE Transactions on
			Information Theory}} {\bfseries 62}, 6816--6830,
	\href{http://arxiv.org/abs/1512.06298}{arXiv:1512.06298},  (2016).
\end{thebibliography}


\appendix

\section{Moderate deviation tail bounds}
\label{app:tailbounds}

\subsection{Lower bound}

Here we apply the lower bound of Ref.~\cite{Rozovsky2002}, which gives a Berry--Esseen-type inequality with multiplicative error.

\begin{lem}[Theorem B2, Ref.~\cite{Rozovsky2002}]
	\label{lem:Rozovsky}
	There exists universal constants $\kappa_1,\kappa_2$ such that, for independent zero-mean variables $X_1,\dots,X_n$ with
	\begin{align}
		V_n:=\frac{1}{n}\sum_{i=1}^n\Var\left[X_i\right]\qquad \text{and}\qquad T_n:=\frac{1}{n}\sum_{i=1}^n\mathbb{E}\left[\left|X_i\right|^3\right],
	\end{align}
	and a $t_n$ bounded
	\begin{align}
		\sqrt{\frac{V_n}{n}}\leq t_n\leq \frac{V_n^2}{T_n},
	\end{align}
	the probability that the average variable $\frac{1}{n}\sum_{i=1}^nX_i$ deviates above the mean by $t_n$ is lower bounded
	\begin{align}
		\ln\Pr\left[\frac{1}{n}\sum_{i=1}^nX_i\geq t_n\right]\geq \ln\Phi\left(-\sqrt{\frac{nt_n^2}{V_n}}\right)-\frac{\kappa_1T_nnt_n^3}{V_n^3}+\ln\left(1-\frac{\kappa_2T_nt_n}{V_n^2}\right).
	\end{align}
\end{lem}
Given this Lemma, we can now prove the desired lower bound on the moderately deviating tail.

\begin{proof}[Proof of \cref{lem:moddev lower}]
	First we note that the bound on the average third absolute moment also imposes a bound on the average variance
	\begin{align}
		V_n
		&=\frac{1}{n}\sum_{i=1}^{n}\mathbb{E}\left[X^2\right]\\
		&<\frac{1}{n}\sum_{i=1}^{n} \mathbb{E}\left[\abs X^3+1\right]\\
		&=T_n+1\\
		&\leq \tau+1.
	\end{align}
	As $\lbrace t_i\rbrace_i$ is moderate, and the moments are bounded $\nu\leq V_n\leq \tau+1$ and $T_n\leq \tau$, there must exist an $N_1(\lbrace t_i\rbrace,\nu,\tau)$ such that
	\begin{align}
		\sqrt{\frac{\tau+1}{n}}\leq t_n\leq \frac{\nu^2}{\tau}
	\end{align}
	for $n\geq N_1$. Applying \cref{lem:Rozovsky}, we have that for $n\geq N_1$
	\begin{align}
		\ln \Pr\left(\frac{1}{n}\sum_{i=1}^{n}X_{i,n}\geq t\right)
		&\geq
		\ln\Phi\left(-\sqrt{\frac{nt_n^2}{V_n}}\right)-\frac{\kappa_1\tau }{\nu^3}nt_n^3+\ln\left(1-\frac{\kappa_2\tau}{\nu^{3/2}}t_n\right).
	\end{align}
	As $nt_n^2\to\infty$ and $V_n\leq \tau+1$, there must exist a constant $N_2(\lbrace t_i\rbrace, \tau)$ such that $n\geq N_2$ implies $nt_n^2/V_n\geq 1$. Using the standard bound $\ln\Phi(-x)\geq -x^2/2-\ln\sqrt{8\pi}x$ for $x\geq 1$, we find
	\begin{align}
		\ln \Pr\left(\frac{1}{n}\sum_{i=1}^{n}X_{i,n}\geq t\right)
		&\geq-\frac{nt_n^2}{2V_n}-\ln\sqrt{8\pi\frac{nt_n^2}{V_n}}
		-\frac{\kappa_1\tau }{\nu^3}nt_n^3+\ln\left(1-\frac{\kappa_2\tau}{\nu^{3/2}}t_n\right).
	\end{align}
	As $t_n$ is moderate, we have that the first term $-nt_n^2/2V_n$ dominates as $n\to\infty$ in the above. As such, for any $\eta>0$, there must exist an $N_3(\lbrace t_n\rbrace, \nu,\tau,\eta)$ such that, for all $n\geq N_3$, the other terms are smaller than this dominant term by a multiplicative factor of $\eta>0$, such that
	\begin{align}
		\ln\sqrt{8\pi\frac{nt_n^2}{V_n}}+\frac{\kappa_1\tau}{V_n^{3}}nt_n^3-\ln\left(1-\frac{\kappa_2\tau}{V_n^{3/2}}t_n\right)\leq \eta\frac{nt_n^2}{2V_n}.
	\end{align}
	We conclude that for $n\geq N(\lbrace t_i\rbrace,\nu,\tau,\eta):=\max\lbrace N_1,N_2,N_3\rbrace$
	\begin{align}
		\ln\Pr\left[\frac{1}{n}\sum_{i=1}^{n}X_{i,n}\geq t_n\right]\geq -(1+\eta)\frac{nt_n^2}{2V_n}.
	\end{align}
\end{proof}

\subsection{Upper bound}

For the upper bound we are going to use a proof technique similar to that used to prove Cram\'er's and Gartner-Ellis theorems in the large deviation regime (see, e.g., Ref.~\cite{dembo98}), and for Lemma 4 of Ref.~\cite{lee16} in the iid moderate deviation regime. However, our approach differs from that in Ref.~\cite{lee16} because we do not want to assume that the average variance, $V_n$, is bounded away from zero.

\begin{proof}[Proof of \cref{lem:moddev upper}]
	Let $h_n(s)$ be the average cumulant generating function $h_n(s)=\frac{1}{n}\sum_{i=1}^n\ln \mathbb{E}\left[e^{sX_{i,n}}\right]$, such that
	\begin{align}
		h_n(0)=0,\qquad\qquad
		h'_n(0)=\frac{1}{n}\sum_{i=1}^n\mathbb{E}\left[X_{i,n}\right]=0,\qquad\qquad
		h''_n(0)=\frac{1}{n}\sum_{i=1}^n\Var\left[X_{i,n}\right]=V_n.
	\end{align}
	For our tail bound we are going to employ a Chernoff bound. Specifically for any $\alpha>0$, an application of the Markov inequality gives
	\begin{align}
		\Pr\left[\frac{1}{n}\sum_{i=1}^nX_{i,n}\geq t_n\right]
		=\Pr\left[e^{\alpha t_n\sum_{i=1}^nX_{i,n}}\geq e^{\alpha nt_n^2}\right]
		\leq\frac{\mathbb{E}\left[e^{\alpha t_n\sum_{i=1}^nX_{i,n}}\right]}{e^{\alpha nt_n^2}}.		
	\end{align}
	Using the independence of $\lbrace X_i\rbrace$, the above bound can be expressed in terms of the average cumulant generating function as
	\begin{align}
		\ln\Pr\left[\frac{1}{n}\sum_{i=1}^nX_{i,n}\geq t_n\right]\leq -n\bigl(\alpha t_n^2-h_n(\alpha t_n)\bigr).
	\end{align} 
	
	In general our choice of $\alpha$ will depend on $n$. If we assume for the moment that $\alpha$ is bounded then, as $t_n\to 0$, there exists a constant $N_1(\lbrace t_i\rbrace,\alpha)$ such that $n\geq N_1$ implies $\alpha t_n\leq 1/2$. Applying Taylor's theorem for such $n$, specifically a second-order expansion with the error in Lagrange form, gives that there exists an $s\in[0,\alpha t_n]\subseteq[0,1/2]$ such that
	\begin{align}
		h_n(s)&=h_n(0)+\alpha t_nh'_n(0)+\alpha^2t_n^2h''_n(0)/2+\alpha^3t_n^3h'''_n(s)/6\\
		&\leq \alpha^2t_n^2V_n/2+\alpha^3t_n^3\gamma/6.
	\end{align}
	
	Plugging this Taylor expansion in to our Chernoff bound above gives
	\begin{align}
		\frac{1}{nt_n^2}\ln\Pr\left[\frac{1}{n}\sum_{i=1}^nX_{i,n}\geq t_n\right]\leq \left(\alpha^2V_n/2-\alpha\right)+\alpha^2\gamma t_n/6.\label{eqn:chernoff}
	\end{align}
	We now need to choose our value of $\alpha$. An obvious choice would be $\alpha=1/V_n$, which gives the tightest possible asymptotic bound. As we have not imposed a lower bound on $V_n$, this value is not necessarily bounded, and therefore could render the previous Taylor expansion invalid. Instead we will slightly modify this choice such that the Taylor expansion is still valid, whilst only changing the final bound by the introducing of an $\eta$. Specifically we will take
	\begin{align}
		\alpha^{-1}:=\sqrt{V_n+\eta/4}\left(\sqrt{V_n+\eta/4}+\sqrt{\eta/4}\right).
	\end{align}
	As required, this choice of $\alpha$ is bounded independent $V_n$ as $\alpha \leq 2/\eta$, meaning that the previous Taylor expansion was indeed valid, and that $N_1=N_1(\lbrace t_i\rbrace,\eta)$. Plugging this choice in to Eq.~\	ref{eqn:chernoff} gives
	\begin{align}
		\ln\Pr\left[\frac{1}{n}\sum_{i=1}^nX_{i,n}\geq t_n\right]
		\leq -\frac{nt_n^2}{2V_n+\eta/2}+\frac{2\gamma}{3\eta^2}nt_n^3.
	\end{align}
	
	Similar to \cref{lem:moddev lower}, the bound on the third derivative of cumulant function bounds the variances as $\Var[X_i]\leq \gamma+1$. Given this, there must exist a constant $N_2(\lbrace t_i\rbrace,\gamma,\eta)$ such that $n\geq N_2$ implies
	\begin{align}
		-\frac{1}{2V_n+\eta/2}+\frac{2\gamma}{3\eta^3}t_n\leq -\frac{1}{2V_n+\eta}.
	\end{align}
	We conclude therefore that for any $n\geq N(\lbrace t_i\rbrace,\gamma,\eta):=\max\lbrace N_1,N_2\rbrace$ we have the desired tail bound
	\begin{align}
		\ln\Pr\left[\frac{1}{n}\sum_{i=1}^{n}X_i\leq t_n\right]\leq -\frac{nt_n^2}{2V_n+\eta}.
	\end{align}
\end{proof}

\subsection{Dimensionless bound}

The non-dimensional bound follows as a corollary of the two previous bound, where we explicitly use the possible dependence of our random variables $X_{i,n}$ on $n$.

\begin{proof}[Proof of \cref{cor:moddev nondim}]
	Starting with random variables $\lbrace X_{i,n}\rbrace _{i\leq n}$, define rescaled variables as $\tilde{X}_{i,n}:=X_{i,n}/\sqrt{V_n}$ for all $i\leq n$. This scaling has the property that it normalises the average variance
	\begin{align}
		\tilde{V}_n:=\frac{1}{n}\sum_{i=1}^{n}\Var[\tilde{X}_{i,n}]=1.
	\end{align}
	As well as this, we can see the dimensionless assumption on $X_{i,n}$
	\begin{align}
		\frac{1}{nV_n^{3/2}}\sum_{i=1}^{n}\sup_{s\in[0,1/2]}\abs{\frac{\mathrm{d}^3}{\mathrm{d}s^3}\ln\mathbb{E}\left[e^{sX_{i,n}}\right]}\leq \gamma,
	\end{align}
	is equivalent to the bound on $\tilde{X}_{i,n}$ 
	\begin{align}
		\frac{1}{n}\sum_{i=1}^{n}\sup_{s\in[0,1/2]}\abs{\frac{\mathrm{d}^3}{\mathrm{d}s^3}\ln\mathbb{E}\left[e^{s\tilde X_{i,n}}\right]}\leq \gamma.
	\end{align}
	Noticing that 
	\begin{align}
		\Pr\left[\frac{1}{n}\sum_{i=1}^{n} X_{i, n}\leq t_n\sqrt{V_n}\right]
		=\Pr\left[\frac{1}{n}\sum_{i=1}^{n} X_{i, n}/\sqrt{V_n}\leq t_n\right]
		=\Pr\left[\frac{1}{n}\sum_{i=1}^{n}\tilde X_{i, n}\leq t_n\right],
	\end{align}
	we can simply apply the existing tail bounds of Lemmas~\ref{lem:moddev lower} and \ref{lem:moddev upper} to $\tilde{X}_{i,n}$, giving that, for any $\eta>0$, there must exist a constant $N(\lbrace t_i\rbrace,\gamma,\eta)$ such that
	\begin{align}
		-(1+\eta)\frac{nt_n^2}{2}\leq\ln\Pr\left[\frac{1}{n}\sum_{i=1}^{n} X_{i, n}\leq t_n\sqrt{V_n}\right]\leq -(1-\eta)\frac{nt_n^2}{2}.
	\end{align}
\end{proof}


\section[Proof of Reversing Lemma]{Proof of \cref{lem:reverse}}
\label{app:reverse}

\begin{lem}
	\label{lem:prereverse}
	Let $A$ and $B$ be two real functions both defined on the same domain
	, with
	\begin{align}
		\hat{A}(b):=\inf_{t} \left\lbrace A(t) \middle| B(t)\leq b\right\rbrace
		\qquad\text{and}\qquad
		\hat{B}(b):=\inf_{t} \left\lbrace B(t) \middle| A(t)\leq a \right\rbrace,
	\end{align}
	then for any $a\geq \inf_{t}A(t)$ and $\delta>0$
	\begin{align}
		\hat{A}\bigl(\hat{B}(a)+\delta\bigr)&\leq a\\
		\hat{A}\bigl(\hat{B}(a)-\delta\bigr)&\geq a.
	\end{align}
\end{lem}
\begin{proof}
	By the definition of the infimum in $\hat{B}(\cdot)$, there must exist a $t^\star$ such that $A(t^\star)\leq a$ and $B(t^\star)\leq \hat{B}(a)+\delta$. Hence we can upper bound
	\begin{align}
		\hat{A}\bigl(\hat{B}(a)+\delta\bigr)=\inf_{s}\left\lbrace A(s) \middle| B(s)\leq \hat{B}(a)+\delta \right\rbrace\leq A(t^\star)\leq a.
	\end{align}
	
	Next, suppose $\hat{A}\bigl(\hat{B}(a)-\delta\bigr)\leq a-\epsilon$ for some $\epsilon>0$. By definition of the infimum in $\hat{A}(\cdot)$, there must therefore exist an $s^\star$ such that $B(s^\star)\leq \hat B(a)-\delta$ and $A(s^\star)\leq a$. This in turn allows us to upper bound
	\begin{align}
		\hat{B}(a)=\inf_{t}\left\lbrace B(t) \middle| A(t)\leq a \right\rbrace\leq B(s^\star).
	\end{align}
	We can therefore conclude that $\hat B(x)\leq \hat B(a)-\delta$, proving $\hat{A}\bigl(\hat{B}(a)-\delta\bigr)> a-\epsilon$ by contradiction. As this is true for arbitrarily small $\epsilon$, we therefore conclude
	\begin{align}
		\hat{A}\bigl(\hat{B}(x)-\delta\bigr)\geq a.
	\end{align}
\end{proof}

\begin{proof}[Proof of \cref{lem:reverse}]
	By swapping both $A_n$ and $B_n$ we can see that the forward and backwards directions of this proof are equivalent, as such we will only consider the forward direction. First, we assume that
	\begin{align}
		\lim\limits_{n\to\infty}\frac{\hat{B}_n(a_n)}{a_n}=1, \quad\forall a_n\text{ moderate}.
	\end{align}
	Next we split the proof of the limit into upper bounding the limit superior, and lower bounding the limit inferior.
	
	Take any moderate sequence $b_n$, and let $a_n:=b_n/2$ and $b_n':=\hat{B_n}(a_n)+b_n/n$. By \cref{lem:prereverse} we have that $\hat{A}_n(b_n') = \hat{A}_n\bigl(\hat{B}_n(a_n)+b_n/n\bigr) \leq a_n$. By assumption we then have that
	\begin{align}
		\lim_{n\to\infty}\frac{b_n'}{b_n}=\lim_{n\to\infty}\frac{\hat{B}_n(a_n)+2a_n/n}{2a_n}=\lim_{n\to\infty}\frac{\hat{B}_n(a_n)}{2a_n}=\frac{1}{2}.
	\end{align}
	As a result we have, for sufficiently large $n$, that $b_n'\leq b_n$. Using this we can bound the limit superior:
	\begin{align}
		\limsup_{n\to\infty}\frac{\hat{A}_n(b_n)}{b_n}
		&\leq \lim_{n\to\infty}\frac{a_n}{b_n'}\\
		&= \lim_{n\to\infty}\frac{a_n}{\hat{B}_n(a_n)}=1. 
	\end{align}
	
	Moreover, if we take $a_n:=2b_n$ and $b_n':=\hat{B}_n(a_n)-b_n/n$ then, by an analogous argument, \cref{lem:prereverse} gives us $\hat{A}_n(b_n')\geq a_n$, and the assumption gives us $b_n'\geq b_n$ for sufficiently large $n$. As such we can also bound the limit inferior 
	\begin{align}
		\liminf_{n\to\infty}\frac{\hat{A}_n(b_n)}{b_n}\geq \lim_{n\to\infty}\frac{a_n}{b_n'}= \lim_{n\to\infty}\frac{a_n}{\hat{B}_n(a_n)}=1. 
	\end{align}
\end{proof}


\section[Proof of Bounded Cumulants]{Proof of \cref{lem:boundedcumulants}}
\label{app:cumulant}

\begin{proof}
	Consider the moment generating function $m(t):=\mathbb{E}\left[e^{tZ}\right]$, such that $h(t)=\ln m(t)$. Similar to the relationship between cumulants and central moments, the derivatives of $h(t)$ can be expressed in terms of derivatives of $m(t)$:
	\begin{align}
		h&=\ln m\\
		h'&=\frac{m'}{m}\\
		h''&=\frac{m''m-m'm'}{m^2}\\
		h'''&=\frac{m'''mm-3m''m'm+2m'm'm'}{m^3}\\
		&~~\vdots
	\end{align}
	As such, if we bound $m(t)$ away from zero, proving that the derivatives of $m(t)$ are uniformly bounded would imply the same about $h(t)$. Noticing that $\sum_{a}r_a=1$ implies $\sum_a r_a^2\geq 1/d$ and $\lambda \leq1/d$, we can see that $m$ is bounded away from zero for any $t\leq 1$:
	\begin{align}
		m(t)
		&=\sum_{a,b}r_a\abs{\braket{\phi_a}{\psi_b}}^2(r_a/s_b)^t\\
		&\geq \frac{1}{\lambda^t}\sum_{a,b}\abs{\braket{\phi_a}{\psi_b}}^2r_a^{t+1}\\
		&\geq \frac{1}{\lambda^t}\sum_{a}r_a^{1+t}\\
		&\geq \frac{1}{\lambda^t}\sum_{a}r_a^2\\
		&\geq \frac{1}{d\lambda}\geq 1
	\end{align}
	
	Next we can use the bound $\sup_{x\in[0,1]}\sqrt{x}\ln^k(1/x)=(2k/e)^k$, to bound the derivatives of the moment generating function for $\abs{t}\leq 1/2$
	\begin{align}
		\abs{m^{(k)}(t)}
		&=\abs{\sum_{a,b}  r_a\abs{\braket{\phi_a}{\psi_b}}^2\left(r_a/s_b\right)^t\ln^k\left(r_a/s_b\right)}\\
		&\leq\sum_{a,b}  r_a\abs{\braket{\phi_a}{\psi_b}}^2\left(r_a/s_b\right)^t\abs{\ln^k\left(r_a/s_b\right)}\\
		&\leq\sum_{a,b}  r_a\abs{\braket{\phi_a}{\psi_b}}^2\left(r_a/s_b\right)^t\left[\ln^k(1/s_b)+\ln^k(1/r_a)\right]\\
		&\leq\sum_{a,b}  \abs{\braket{\phi_a}{\psi_b}}^2\sqrt{r_a/s_b}\left[\ln^k(1/s_b)+\ln^k(1/r_a)\right]\\
		&\leq
		\max_{\substack{\lambda\leq s\leq1\\0\leq r\leq 1}}
		\sqrt{r/s}\left[\ln^k(1/s)+\ln^k(1/r)\right]\\
		&\leq\frac{1}{\sqrt{\lambda}}\max_{\substack{\lambda\leq s\leq1\\0\leq r\leq 1}}\left[\ln^k(1/s)+\sqrt{r}\ln^k(1/r)\right]\\
		&\leq\frac{1}{\sqrt{\lambda}}\left[\ln^k\left(\frac{1}{\lambda}\right)+\left(2k/e\right)^k\right]=:C_k.
	\end{align}
\end{proof}

\end{document}